\documentclass[11pt]{article}

\usepackage[letterpaper,margin=1in]{geometry}
\usepackage{amssymb}
\usepackage{amsmath}
\usepackage{amsthm}
\usepackage{hyperref}
\usepackage{tikz}
\usetikzlibrary{tikzmark,arrows.meta}
\usepackage{physics}
\usepackage{multirow}
\usepackage{cite}
\usepackage{caption}

\newtheorem{theorem}{Theorem}[section]
\newtheorem{corollary}[theorem]{Corollary}
\newtheorem{lemma}[theorem]{Lemma}

\newtheorem{remark}{Remark}

\newcommand{\p}{\mathbb{P}}
\newcommand{\R}{\mathbb{R}}
\newcommand{\cO}{\mathcal{O}}
\newcommand{\rb}{{r_\star}}
\newcommand{\one}{{\mathbf{1}}}
\newcommand{\zero}{{\mathbf{0}}}

\usepackage{algorithm}
\usepackage{algpseudocode}
\algrenewcommand\algorithmicindent{2em}

\DeclareMathOperator{\sign}{sign}
\DeclareMathOperator{\diag}{diag}

\title{On the Complexity of Low-Rank Matrix Signing and \\  Entrywise Power Matrix Factorization}

\author{Nicolas Gillis\thanks{Department of Mathematics and Operational Research. We acknowledge  the support by the European Union (ERC consolidator, eLinoR, no 101085607).  
Emails: firstname.lastname@umons.ac.be.} \and Subhayan Saha\footnotemark[1] \and Stefano Sicilia\footnotemark[1] $^,$\thanks{SS is a member of the Gruppo Nazionale Calcolo Scientifico-Istituto Nazionale di Alta Matematica (GNCS-INdAM).} \and Arnaud Vandaele\footnotemark[1]}

\date{University of Mons, Mons, Belgium}

\begin{document}
\maketitle



\begin{abstract}
Given a nonnegative matrix $X$, a factorization rank $r$ and {a positive integer $p$}, entrywise power matrix factorization (EPMF) looks for a low-rank matrix $X_r$ such that $X = |X_r|^{\circ p}$ (exact case) or $X \approx |X_r|^{\circ p}$ (approximate case), where $(\cdot)^{\circ p}$ denotes the componentwise exponent. EPMF includes the modulus model ($p=1$) and componentwise square factorization ($p=2$) as special cases, the latter being closely related to the square root rank. We analyze the computational complexity of  the exact decision problem and the Frobenius-norm approximation problem, and establish a complete complexity landscape. In the exact case, we show that EPMF is equivalent to the combinatorial problem of flipping the signs of the entries of a given matrix $X$ to obtain a rank-$r$ matrix, which we refer to as the low-rank matrix signing (LRMS) problem. We first show that LRMS, and hence exact EPMF, is strongly NP-hard, improving a weak NP-hardness result for the square-root-rank (Math. Prog., 2015). We then show that LRMS can be solved in polynomial time when $r$ is fixed. Moreover, when the rank $r$ is part of the input, we show that for generic matrices the algorithm is fixed-parameter tractable (FPT) in the parameter $r$; in fact, the running time is fixed-parameter linear in the number of entries of the input matrix. In the approximate case using the Frobenius norm as an error measure, we show that EPMF is NP-hard, already when $r=2$, the smallest nontrivial case. 
\end{abstract}



\section{Introduction}\label{sec:intro}

Low-rank matrix approximation is a central tool in numerical linear algebra, data analysis, machine learning, and signal processing. Given a data matrix $X\in\mathbb R^{m\times n}$ and an integer $r \leq \min(m,n)$, the classical problem is to approximate $X$ by a matrix of rank at most $r$, or, equivalently, to find factors $W\in\mathbb R^{m\times r}$ and $H\in\mathbb R^{r\times n}$ such that $X\approx WH$. This is equivalent to linear dimensionality reduction: each column of $X$ is approximated by a linear
combination of the columns of~$W$.
The best known example is the truncated singular value decomposition (SVD), which gives an optimal rank-$r$ approximation in the Frobenius norm, and any unitary invariant norm, by the Eckart--Young theorem~\cite{eckart1936approximation, golub2013matrix, trefethen2022numerical}.    
Many variants, including robust PCA, sparse PCA, weighted low-rank approximation, and matrix completion, modify the loss function or impose additional constraints in order to reflect the structure of the data; see, e.g., \cite{udell2016generalized} and the references therein. 
The computational complexity of such models has been studied in depth.  
Computing the (truncated) SVD is a classical problem in numerical linear algebra, and can be done in polynomial time in $m, n, r$ and $\log(1/\varepsilon)$, where $\varepsilon$ is the desired precision; see, e.g.,~\cite{golub2013matrix, trefethen2022numerical}.   
The SVD can be used to compute the best rank-$r$ approximation in any unitary invariant norm, including the matrix $\ell_2$ norm and the Frobenius norm. However, as soon as different norms are used, finding the best rank-$r$ approximation becomes NP-hard, already when $r=1$; this is true for the componentwise $\ell_1$~\cite{gillis2018complexity} and $\ell_\infty$ norms~\cite{gillis2019low}, and weighted norms or when data is missing~\cite{gillis2011low}. 
Moreover, by writing $X_r = WH$, where $W$ has $r$ columns and $H$ has $r$ rows, it is often useful in practice to add constraints on the factors, $W$ and $H$, typically to improve the interpretability of the decomposition. For example, sparsity leads to sparse PCA and variants, and imposing a given sparsity pattern to the factors leads to an NP-hard problem~\cite{le2023spurious}. 
Another famous example are nonnegativity constraints leading to nonnegative matrix factorization (NMF) which is also NP-hard~\cite{vavasis2010complexity}, 
but can be solved in polynomial time in the exact case (that is, find, if possible, $W \geq 0$, $H \geq 0$ such that $X = WH$) when $r$ is treated as a constant, with a complexity of $(mn)^{\mathcal{O}(r^2)}$ operations~\cite{arora2012computing}.  
Analogous results were obtained for the positive semidefinite factorization~\cite{GRT13,shitov2018matrices}.
These results study worst-case complexity, but luckily there is a plethora of low-rank models that are tractable under reasonable assumptions. 
Two of the most widely known approaches to obtain such results are based on 
minimizing the nuclear norm as a proxy for the rank~\cite{recht2010guaranteed}, or on proving that the optimization landscape does not have spurious local minima, that is, all local minima are global~\cite{ge2017no, chi2019nonconvex}. 
This has been done for missing data~\cite{candes2010power, ge2016matrix}, robust variants that rely on other metrics than the Frobenius norm~\cite{chandrasekaran2011rank, candes2010robust}, NMF under the so-called separability assumption~\cite{arora2012computing}, and 
binary matrix decompositions~\cite{kueng2021binary, kueng2019binary}, to cite a few. 
It is out of the scope of this paper to review this literature.

Despite their success, linear low-rank models have limited expressive power.
Many data sets exhibit nonlinear structures that they cannot capture. 
This limitation has motivated the recent introduction of {nonlinear matrix decompositions} (NMDs)~\cite{saul2022nonlinear}, in which one approximates a matrix by applying a nonlinear function entrywise to a low-rank 
matrix $X_r$:  
\[
X\approx f(X_r). 
\]
Different choices of $f$ lead to different models and allow for various structural assumptions about the data. 
A prominent example is the ReLU model, 
$X \approx \max(0,X_r)$, 
introduced by Saul for sparse nonnegative data~\cite{saul2022nonlinear}. 
{Since then, several algorithmic approaches have been proposed for this model,} 
including
coordinate-descent schemes~\cite{awari2024coordinate}, 
accelerated alternating partial Bregman proximal gradient method~\cite{wang2025accelerated}, 
and extrapolated block coordinate methods~\cite{gillis2025extrapolated}. 
Another example is the componentwise square factorization (CSF):
$X\approx (WH)^{\circ 2}$. 
CSF is closely related to the square-root rank of nonnegative matrices, a quantity that appears in the compact representation of 
convex polytopes~\cite{lee2014square,fawzi2015positive}. 
CSF is also related to the {low-rank matrix signing} (LRMS) problem: given a matrix $M \in \R^{m \times n}$ and an integer $r \leq \min(m,n)$, decide whether there exists  a sign matrix $S \in \{\pm 1\}^{m \times n}$ such that $\rank(S \circ M) \leq r$, {where $\circ$ is the entrywise product.} 
LRMS is connected to the {signless rank} of a matrix, denoted by $\rank_{\pm}(A)$, which is defined as the minimum rank of $S \circ A$ over all sign matrices $S \in \{\pm 1\}^{m \times n}$~\cite{goucha2021phaseless}. We will study this problem in Sections~\ref{sec:strongNPhardness}~and~\ref{sec:exact-fixed-rank}. 
A natural generalization of the signless rank for complex matrices is referred to as the {phaseless rank} of a matrix \cite{goucha2021phaseless}, which is connected to the notion of equimodular classes of matrices \cite{Camion1966-wt}. 
CSF has also been used to obtain compact representations of nonnegative data and probabilistic circuits~\cite{loconte2024subtractive}, 
and recent work has developed coordinate-descent algorithms for the associated least-squares problem~\cite{lefebvre2024component}.  
CSF is also closely related to the Hadamard decomposition, defined as $X \approx X_1 \circ X_2$, where $X_i$ are low-rank matrices, that has been successfully used for example to compress or adapt neural networks~\cite{FedPara, huang2025hira}; see~\cite{gillis2026manifold} and the references therein for more details. 
The modulus model, $X\approx |WH|$, 
provides another unconstrained factorization model for nonnegative data~\cite{awari2025alternating}. 
Other nonlinearities, such as min-max and sigmoid functions, have also been considered in applications with bounded or binary observations~\cite{awari2025alternating, nguyen2025nonlinear}.

\subsection{Contribution of the paper} 

The computational complexity of (linear) low-rank matrix approximations has been extensively studied; see Section~\ref{sec:intro}.  
However, for non-linear models, computational complexity has not been explored yet. In this paper, we focus on entrywise power matrix factorization (EPMF), which unifies the modulus and componentwise square cases.
For a fixed {integer $p\geq 1$}, it considers the decomposition 
$X \approx |X_r|^{\circ p}$, where $X_r$ is a rank-$r$ matrix, and the absolute value and the power are applied entrywise. We consider two variants: given a nonnegative matrix $X$, a factorization rank $r$ and {an integer $p\geq 1$,} 
\begin{enumerate}
    \item \textsc{ExactEPMF}: compute, if possible, a rank-$r$ matrix $X_r$ such that $X = |X_r|^{\circ p}$. 

    \item \textsc{FroEPMF}: find a rank-$r$ matrix $X_r$ that minimizes $\left\|X - |X_r|^{\circ p}\right\|_F^2$.
\end{enumerate}

Our positive result concerns the exact fixed-rank regime.
After taking the entrywise $p$th root of~$X$, the problem \textsc{ExactEPMF} becomes an LRMS problem: given a nonnegative matrix $M$, decide whether one can choose the signs  of its entries so that the resulting matrix has rank at most $r$. 
We show that, for fixed $r$, this problem can be solved in polynomial time. As for exact NMF, the problem is polynomial-time solvable for fixed rank, although the exponent of the polynomial in the dimension of the input depends on $r$.  
The algorithm enumerates all possible sign patterns of non-singular basis $r\times r$ blocks, but it avoids a global enumeration of the signs by using the linear dependencies. 
Moreover, even when $r$ is given as part of the input, if the input matrix is generic, we show that the algorithm is fixed-parameter tractable (FPT); that is, it requires $f(r)\operatorname{poly}(m,n)$ arithmetic operations, where $f$ is a computable function. 

Our hardness results show that this favorable behavior is specific to exact fixed-rank feasibility or generic settings.
When the rank is part of the input, we show that \textsc{ExactEPMF} is strongly NP-hard, by a reduction from monotone not-all-equal 3-SAT, generalizing and improving a weak NP-hardness result for the square-root-rank case given by the authors of~\cite{fawzi2015positive}. 
Moreover, we show that the least-squares formulation, \textsc{FroEPMF}, is NP-hard, already for $r=2$, via a reduction from the decision version of the \textsc{Cut-Norm} problem.



 Table~\ref{tab:power-complexity} summarizes our results. 

\begin{table}[h]
\centering
\renewcommand{\arraystretch}{1.25}
\resizebox{\textwidth}{!}{
\begin{tabular}{c|c|c}
& \textsc{ExactEPMF} & \textsc{FroEPMF} \\
\hline
\multirow{2}{*}{$r$ part of the input}
& Strongly NP-hard (Section~\ref{sec:strongNPhardness})
& \multirow{2}{*}{\tikzmarknode{froHARD}{{Strongly} NP-hard}}
\\
& FPT in the generic case (Section~\ref{sec:FPTgeneric}) & \\
\hline
$r$ fixed &
Polynomial time algorithm (Section~\ref{sec:exact-fixed-rank})
& \tikzmarknode{fixedHARD}{{Strongly} NP-hard for $r=2$ (Section~\ref{sec:fro-fixed-rank})} 
\\ 
\end{tabular}
\begin{tikzpicture}[remember picture,overlay]
\draw[-{Stealth[length=5pt]},thick]
  ([xshift=2pt,yshift=-1pt]fixedHARD.east)
  to[out=0,in=0,looseness=1.6]
  ([xshift=2pt,yshift=1pt]froHARD.east);
\end{tikzpicture}}
\caption{Contributions of this paper:  complexity landscape for EPMF.}
\label{tab:power-complexity}
\end{table}

\subsection{Proof overview and algorithmic ideas}


We give a more detailed overview of our results and the techniques used to achieve them.

\paragraph{ExactEPMF.}

Our results on \textsc{ExactEPMF} are established via exploring the connections with the equivalent LRMS, namely deciding whether a nonnegative matrix admits a signing whose rank is at most $r$. Indeed, after taking entrywise $p$th roots, the magnitudes of the entries are fixed, and the only remaining freedom is to assign signs so that the resulting matrix has rank at most~$r$. 

\begin{enumerate}
    \item For proving the strong NP-hardness (Theorem \ref{thm:satiffsignrank} in Section~\ref{sec:strongNPhardness}), we reduce \textsc{ExactEPMF} from the \textsc{Monotone NAE-3SAT} whose instance is given by 
    \[
     \phi = \bigwedge_{l=1}^m \text{NAE}(x_{i_l}, x_{j_l}, x_{k_l}), 
    \]
    with $n$ Boolean variables $x_1,\dots,x_n$, and $m$ clauses, where $\text{NAE}(x_{i_l}, x_{j_l}, x_{k_l})$ means that $x_{i_l}, x_{j_l}, x_{k_l}$ cannot be all equal.  
    The reduction constructs an $(n+m+3) \times (n+4)$ matrix with entries in $\{0,1,2,5\}$. 
    The first $n+3$ rows of the construction contain a submatrix that has rank $n+3$, 
    regardless of the signing chosen for its entries, while the remaining $m$ rows correspond to the clauses and 
    are constructed such that the matrix has a signing of rank $n+3$ {if and} only if there exists     
    an assignment $x$ satisfying the not-all-equal conditions. 
    
    \item For a fixed $r$ and $M=X^{\circ 1/p}$, a naive method to solve \textsc{ExactEPMF} would be the following algorithm for the equivalent \textsc{LRMS} problem for $M$ - try all possible sign matrices $S\in \{\pm 1\}^{m\times n}$ and  to check whether $S\circ M$ has rank $r$, leading to a computational cost larger than $\cO(2^{mn})$. In Section~\ref{sec:exact-fixed-rank}, we instead propose a tractable algorithm (Algorithm~\ref{alg:improvedfinal}) that determines if $M$ can be written as $M=|Y|$, where $Y\in \R^{m\times n}$ is a rank-$k$ matrix for some constant $k \leq r$ in $\operatorname{poly}(m,n)$ running time (Theorem \ref{thm:fixed-rank-signing}). The algorithm exploits the fact that the matrix $Y$ has rank $k$ if and only if it contains a nonsingular $k \times k$ submatrix $Y(I,J)$. This submatrix determines the following {skeleton decomposition} 
    \[
        Y = Y(:,J) Y(I,J)^{-1} Y(I,:).
    \]
    We go over all the possible subsets $I$ and $J$ and reduce the problem to check whether there is a signing of the candidates $W:=M(:,J)$ and $H:=Y(I,J)^{-1} Y(I,:)$ such that $M=|WH|$. 
    In order to verify the validity of these candidates, we check if they satisfy the associated nonlinear constraints of the form $|W(i,:) H(:,j)| = M(i,j)$, which we rewrite as linear constraints on $k \times k$ symmetric matrices: $\langle W(i,:)^\top W(i,:), H(:,j) H(:,j)^\top \rangle = M(i,j)^2$. 
    Verifying all such constraints would still be exponential in the dimensions of the matrix, but by using the crucial observation that the space of  symmetric matrices has dimension $\binom{k+1}{2}$, we show that it is enough to explore only a small subset of these constraints. 
    
    \item In Section~\ref{sec:FPTgeneric}, we prove that, for generic inputs of given rank $r$, Algorithm~\ref{alg:improvedfinal} leads to a stronger running-time guarantee $2^{\cO(r^3)}\cdot \cO(mn)$ (see Theorem~\ref{thm:FPTgeneric}). We use the following two observations:
    \begin{itemize}
        \item For any $k\leq r$, every $k \times k$ submatrix of a generic rank-$r$ matrix is nonsingular. 
        \item For any generic $N_k = \binom{k+1}{2}$ vectors $\{w_i\}_{i\in[N_k]}$, the matrices $\{w_iw_i^\top\}_{i \in [N_k]}$ span the entire space of symmetric matrices.
    \end{itemize}
    The remaining search depends only on $r$, while the number of steps required to process the matrix is linear in $mn$. 
\end{enumerate}

\paragraph{FroEPMF.} To prove the strongly NP-hardness of the approximate EPMF in the Frobenius norm for $r=2$, we rely on a reduction from the \textsc{Cut-Norm} decision problem. 
Formally, 
\begin{itemize}
    \item given a matrix $M\in\{\pm 1\}^{s\times t}$ and $D\in\mathbb{N}$, the \textsc{Cut-Norm} decision problem asks whether there exist $(u,v)\in\{\pm 1\}^s\times\{\pm 1\}^t$ such that $\langle M,uv^\top\rangle \geq D$;
    \item given a matrix $X\in\mathbb{R}_{\geq 0}^{m\times n}$, a fixed $p\geq 1$, and $T>0$ the \textsc{FroEPMF} problem consists of deciding whether there exist $(W,H)\in\mathbb{R}^{m\times 2}\times\mathbb{R}^{2\times n}$ such that $\|X-|WH|^{\circ p}\|_F^2\leq T$.
\end{itemize}
Given a cut-norm instance, the technical reduction constructs a block matrix $X$ of size $(s+2N)\times(t+2N)$ with entries in $\{0,1,2^p\}$, for a sufficiently large $N=N(s,t,p)$. In Theorem~\ref{thm:red_CUT_to_R2FroEPMF}, we show that the given cut-norm instance is equivalent to this instance of the rank-$2$ \textsc{FroEPMF} problem with $T=\frac{st-D}{2}$.

\subsection{Outline and notation of the paper}

The paper is organized as follows. 
Section~\ref{sec:exact-nonfixed-rank} proves strong NP-hardness of \textsc{ExactEPMF} when the rank $r$ is part of the input. 
When the rank $r$ is not part of the input, meaning it is treated as a constant, 
we provide in Section~\ref{sec:exact-fixed-rank} a  polynomial-time algorithm
for \textsc{ExactEPMF}. Moreover, in Section~\ref{sec:FPTgeneric} we show that this algorithm is fixed-parameter tractable (FPT) for generic input matrices even when $r$ is a part of the input. 
Section~\ref{sec:fro-fixed-rank} proves strong NP-hardness of \textsc{FroEPMF}, already when $r=2$, the first non-trivial case. Section~\ref{sec:concl} draws the conclusions and the appendix includes some technical lemmas.

\paragraph{Notation}

 For the first $n$ natural numbers, we use the standard notation $[n]=\{1,\dots,n\}$. The vectors $e_1,\dots,e_n$ denote the canonical basis of $\R^{n}$. For a matrix $A$, we use a \textsc{Matlab}-like notation: we denote its $i$th row by $A(i,:)$, its $j$th column by $A(:,j)$ and its $(i,j)$ entry by $A(i,j)$, sometimes also indicated with the lowercase associated letter $a_{ij}$. We use $\circ$ for the entrywise multiplication, both for matrices and vectors, but also for the $p$th entrywise power of a matrix $A$, which is denoted by $A^{\circ p}$. The transpose of $A$ is $A^\top$. For any vector $v$, $\diag(v)$ is the diagonal matrix whose diagonal is $v$, and $\|v\|=\sqrt{\langle v,v\rangle}$ is the $\ell_2$-norm of $v$. The all-zeros matrix and the all-ones matrix of size $m\times n$ are denoted, respectively, by $\mathbf{0}_{m\times n}$ and $\mathbf{1}_{m\times n}$, and in case of the vectors we write $\mathbf{0}_n$ and $\mathbf{1}_n$. The $n\times n$ identity matrix is $I_n$.

\section{NP-hardness of ExactEPMF}\label{sec:exact-nonfixed-rank}

We first consider \textsc{ExactEPMF} when the rank is part of the input.
We first generalize the result of~\cite{fawzi2015positive}, which proves the weak NP-hardness of computing
the square-root rank, that is, \textsc{ExactEPMF} for $p=2$, using a reduction from \textsc{Partition} which is weakly NP-hard~\cite{garey2002computers}; see Section~\ref{sec:weakNPhardness}. 
We then prove strong NP-hardness in Section~\ref{sec:strongNPhardness}. 

\subsection{Weak NP-hardness}\label{sec:weakNPhardness}

The reduction below is a direct adaptation of the construction in~\cite{fawzi2015positive} from \textsc{Partition}, which we provide here for completeness. Let us formally define the decision versions of both problems. 

\noindent
\fbox{\begin{minipage}{0.99\textwidth}\textsc{Partition.} 

\noindent
\emph{Given:} a multi-set\footnotemark~of positive integers $\mathcal{S}=\{s_1,...,s_N\}$.  

\noindent
\emph{Question:}
does there exist a partition of $\mathcal{S}$ into two sub multi-sets $\mathcal{S}_1$ and $\mathcal{S}_2$ such that the sum of the numbers in $\mathcal{S}_1$ equals the sum of the numbers in $\mathcal{S}_2$?
\end{minipage}}\\

\footnotetext{By multi-set we mean a collection of elements that can possibly contain several times the same element.}

\vspace{0.1cm}

\noindent
\fbox{\begin{minipage}{0.99\textwidth}\textsc{ExactEPMF.}

\noindent
\emph{Given:} a matrix $X\in\mathbb R_{\ge 0}^{m\times n}$ and a positive integer $r$. 

\noindent
\emph{Question:} for a fixed integer $p\geq 1$, 
do there exist matrices $W\in\mathbb R^{m\times r}$ and $H\in\mathbb R^{r\times n}$ such that $X=|WH|^{\circ p}$?

\end{minipage}}\\
Let $\mathcal{S}$ be an instance of \textsc{Partition}. We construct an instance $(X,r)$ of \textsc{ExactEPMF} as follows. Set $m = n = N+1$ and $r=N$. Using $s=(s_1,\ldots,s_N)^\top$, set 
    $$X=
        \begin{pmatrix}
        I_N & s^{\circ p}\\
        \mathbf 1_N^\top & 0
        \end{pmatrix}
        \in\mathbb R_{\ge0}^{(N+1)\times(N+1)}.
        $$

The construction is polynomial in the size of the input since $p$ is fixed.  
Intuitively, the reduction relies only on the freedom of choosing signs of prescribed entrywise $p$th roots.   
\begin{theorem}\label{thm:exactEPMF-weakly-nphard}
The instance $(X,r)$ is a yes-instance of \textsc{ExactEPMF} if and only if $\mathcal{S}$ is a yes-instance of \textsc{Partition}.
Hence, \textsc{ExactEPMF} is weakly NP-hard when $r$ is part of the input.
\end{theorem}

\begin{proof}
\textbf{If part.}
Suppose that $\mathcal{S}$ is a yes-instance of \textsc{Partition}.
Then there exist signs $\sigma_i\in\{\pm1\}$ such that $\sum_{i=1}^N \sigma_i s_i=0$. 
Define
\[
Y=\begin{pmatrix}
I_N & s\\
\sigma^\top & 0
\end{pmatrix}.
\]
Then $|Y|^{\circ p}=X$.
The last row of $Y$ is the linear combination of the first $N$ rows with coefficients $\sigma_1,\ldots,\sigma_N$ since $\sum_{i=1}^N \sigma_i s_i=0$.
Thus $\rank(Y)\leq N=r$. 
Therefore $X=|WH|^{\circ p}$ for some
$W\in\mathbb R^{(N+1)\times N}$, $H\in\mathbb R^{N\times(N+1)}$, so the \textsc{ExactEPMF} instance is a yes-instance. 

\textbf{Only if part.}
Suppose that the \textsc{ExactEPMF} instance is a yes-instance.
Then there exists a matrix $Y$ such that
$|Y|^{\circ p}=X$, $\rank(Y)\leq N$.
Since the upper-left block of $X$ is $I_N$, the corresponding block of $Y$ is a signed identity matrix, hence is invertible.
Therefore $\operatorname{rank}(Y)\geq N,$ and so $\operatorname{rank}(Y)=N$.
The first $N$ columns of $Y$ are linearly independent, and the last column must lie in their span, that is, 
$Y(:,N+1)=\sum_{i=1}^N \alpha_i Y(:,i)$ for some $\alpha$'s.  
The first $N$ entries of this relation indicate that $\alpha_i=\pm s_i$ while the last one gives 
$\sum_{i=1}^N \sigma_i s_i=0$  
for some signs $\sigma_i\in\{\pm1\}$.
Hence, the \textsc{Partition} instance is a yes-instance.
\end{proof}

\subsection{Strong NP-hardness}\label{sec:strongNPhardness}

To prove strong NP-hardness of \textsc{ExactEPMF}, we rely on the Monotone Not-All-Equal 3-Satisfiability problem (NAE-3SAT). 

For Boolean variables, $x_1,\dots,x_n\in\{0,1\}$,   $\text{NAE}(x_1,\dots,x_n)$ is satisfied by the assignment $x$ 
if and only if there exists $i \neq j \in [n]$ such that $x_i$ is $1$ and $x_j$ is $0$. We say that a map $\phi: \{0,1\}^n \xrightarrow[]{} \{0,1\}$ is an instance of the Monotone NAE-3SAT problem if it can be written as 
\begin{equation} \label{eq:nae3sat}
   \phi = \bigwedge_{l=1}^m \text{NAE}(x_{i_l}, x_{j_l}, x_{k_l}) 
\end{equation}
over Boolean variables $x_1, \dots, x_n$. Monotone \textsc{NAE-3SAT} is known to be strongly NP-complete, even under different restrictions \cite{Schaefer78,DARMANN2020}. The following is the decision version of the problem:

\noindent
\fbox{\begin{minipage}{0.99\textwidth}\textsc{Monotone NAE-3SAT.}

\noindent
\emph{Given:} a Monotone NAE-3SAT instance $\phi$ over Boolean variables $x_1, \dots, x_n$. 

\noindent
\emph{Question:}
does there exist an assignment $x_1, \dots, x_n \in \{0,1\}^n$ such that $\phi(x_1, \dots,x_n)$ is satisfied?
\end{minipage}}\\




\vspace{0.25cm}

To prove strong NP-hardness of \textsc{ExactEPMF}, we rely on the following intermediate combinatorial problem. 

\noindent
\fbox{\begin{minipage}{0.99\textwidth}\textsc{Low-Rank Matrix Signing (LRMS)}

\noindent
\emph{Given:} a nonnegative integer matrix $M\in\mathbb{R}_{\geq 0}^{s\times t}$ and an integer $r\leq \min(s,t)$. 

\noindent
\emph{Question:}
does there exist a sign matrix $S\in \{\pm 1\}^{s\times t}$ such that
$\rank(S\circ M)\leq r$?
\end{minipage}}\\

It remains open whether the same hardness result holds for binary input matrices.

Given a nonnegative matrix $M$, LRMS requires flipping some signs of the entries of $M$ in order to make its rank at most $r$.
This problem is equivalent to \textsc{ExactEPMF} on instances of the form $X=M^{\circ p}$, where $p\geq 1$ is a fixed integer.
Indeed, if $S\circ M=WH$ with $\rank(WH)\leq r$, then
\[
X=M^{\circ p}=|S\circ M|^{\circ p}=|WH|^{\circ p},
\]
so $(X,r)$ is a yes-instance of \textsc{ExactEPMF}. 
Conversely, if $X=M^{\circ p}=|WH|^{\circ p}$, then $|WH|=M$, and hence $WH=S\circ M$ for the sign matrix $S=\operatorname{sign}(WH)$, with $\rank(S\circ M)\leq r$.
Therefore, it suffices to prove strong NP-hardness of LRMS. The corresponding hardness result for \textsc{ExactEPMF} then follows immediately by mapping $M$ to $M^{\circ p}$.

\paragraph{The reduction.}
For a Monotone NAE-3SAT instance $\phi = \land_{l=1}^m \text{NAE}(x_{i_l}, x_{j_l}, x_{k_l})$ over $n$ variables $x_1, \dots, x_n$, we define the $m \times n$ matrix $H_{\phi}$ with the $\ell$th row corresponding to the $\ell$th $\text{NAE}$ clause given by $(e_{i_\ell} + e_{j_\ell} + e_{k_\ell})^\top$. 
We then construct the following $(m+n+3) \times (n+4)$ matrix 
\begin{equation}\label{eq:defninstanceAphi}
   M_\phi
=
\left(
\begin{array}{cccc|c}
1 & 1 & 0 & 0 & \mathbf{0}_{1\times n}\\
1 & 0 & 1 & 0 & \mathbf{0}_{1\times n}\\
5 & 0 & 0 & 1 & \mathbf{0}_{1\times n}\\
\hline
\mathbf{0}_{n\times 1}
&
\mathbf{1}_{n}
&
2\mathbf{1}_{n}
&
\mathbf{0}_{n\times 1}
&
I_n\\
\hline
\mathbf{0}_{m\times 1}
&
\mathbf{0}_{m\times 1}
&
\mathbf{0}_{m\times 1}
&
\mathbf{1}_{m}
&
H_\phi
\end{array}
\right). 
\end{equation} 
Note that for all sign patterns $S \in \{\pm 1\}^{(m+n+3) \times (n+4)}$, $\rank(S \circ M_\phi) \in \{n+3, n+4\}$. 
Indeed consider the submatrix with the last $n+3$ columns and the first $n+3$ rows: 
\[
S(1:n+3,2:n+4)\circ \left(\begin{array}{c|c}
I_3 & 0 \\
\hline
* & I_n
\end{array}\right),\]
which has a non-zero determinant regardless of the entries of $S$, hence $\rank(S \circ M_\phi) \in \{n+3,n+4\}$.

We can now prove the following result. 
\begin{theorem}\label{thm:satiffsignrank}
    The instance $\phi$ is satisfiable if and only if there exists a sign pattern $S$ such that $\rank(S \circ M_{\phi}) = n+3$. 
    \begin{proof}
        \textbf{If part.} Any assignment $x$ of $\phi$ uniquely determines
        \[
         \tilde z(x)=
         \begin{pmatrix}
             1 \\
             1 \\
             1 \\
             5 \\ 
             z \\             
         \end{pmatrix},
         \qquad
         S_\phi(x)=
         \left(
         \begin{array}{cccc|c}
            1 & -1 & 1 & 1 & \one_n^\top \\
            1 & 1 & -1 & 1 & \one_n^\top \\
            1 & 1 & 1 & -1 & \one_n^\top \\
            \hline
            \one_n & \one_n & -v & \one_n & \diag(v) \\
            \hline
            \one_m & \one_m & \one_m & -\one_m & \Sigma \\
         \end{array}
         \right),
        \]
        where $z=z(x)=2x+\one_n$, $v=v(x) = \mathbf{1}_n-2x \in \{\pm 1\}^n$ and $\Sigma=\Sigma_\phi(x)\in \{\pm 1\}^{m\times n}$ is such that  
        \[
         \Sigma(\ell,[i_\ell,j_\ell,k_\ell])=[f(x_{i_\ell},\sigma_{\ell}),f(x_{j_\ell},\sigma_{\ell}),f(x_{k_\ell},\sigma_{\ell})] \; \text{ for all $\ell\in[m]$},
        \]
        with
        \[
        f(y,\sigma)=y- (-1)^\sigma(1-y)\in\{\pm1\},\ \text{ where } y \in \{0,1\}  \text{ and }  \sigma_{\ell}=x_{i_\ell}+x_{j_\ell}+x_{k_\ell}, 
        \] 
        and the other entries of $\Sigma$ are  equal to $1$. Note that $\sigma_\ell \in \{1,2\}$ for an instance 
        to be satisfied.
        One can check that $x$ satisfies $\phi$ if and only if $\tilde z=\tilde z(x)$ is in the kernel of 
        \[
         M_\phi \circ S_\phi(x)=
         \left(
         \begin{array}{cccc|c}
            1 & -1 & 0 & 0 & \zero_n^\top \\
            1 & 0 & -1 & 0 & \zero_n^\top \\
            5 & 0 & 0 & -1 & \zero_n^\top \\
            \hline
            \zero_n & \one_n & -2v & \zero_n & \diag(v) \\
            \hline
            \zero_m & \zero_m & \zero_m & -\one_m & \Sigma\circ H_\phi \\
         \end{array}
         \right).
        \]
        Let us denote 
        \begin{equation}
            \label{eq:abc}
            \begin{pmatrix} 
              a \\
              b \\
              c \\
            \end{pmatrix}=(M_\phi\circ S_\phi)\tilde{z}, \qquad a\in \R^3, \ b\in \R^n, \ c\in \R^m.
        \end{equation}
        Let $x$ be an assignment that satisfies $\phi$. The vector $a$ is zero by the choice of $\tilde{z}$. For every entry of $b$, 
        \[
         b_i=1-2 v_i+v_i \tilde{z}_{i+4}=1-2(-2x_i+1)+(-2x_i+1)(2x_i+1)=1+4x_i-2-4x_i+1=0,
        \]
        where we used $x_i^2=x_i$. Moreover, for every $\ell\in[m]$, 
        \begin{align*}
            c_\ell & = -5+\sum_{h\in\{i,j,k\}}(2x_{h_\ell}+1)(x_{h_\ell}-(-1)^{\sigma_\ell}(1-x_{h_\ell}))\\
            & = -5+\sum_{h\in\{i,j,k\}} 3x_{h_\ell}-(-1)^{\sigma_\ell}+(-1)^{\sigma_\ell} x_{h_\ell} =-5+3{\sigma_\ell} -3(-1)^{\sigma_\ell}+(-1)^{\sigma_\ell} {\sigma_\ell},
        \end{align*}
        which vanishes if and only if $\sigma_\ell\in \{1,2\}$, that are exactly the values corresponding to an assignment $x$ that satisfies $\phi$. 
        
        \textbf{Only if part.} Suppose that $\rank(S\circ M_\phi)=n+3$ for some
$S\in\{\pm1\}^{(m+n+3)\times(n+4)}$.
Let $\tilde{z}\neq 0$ span the one-dimensional kernel of $S\circ M_\phi$.
We may assume
w.l.o.g.\ that $\tilde z\geq0$ since flipping the signs of the $j$th 
column of $S$ and of the
entry $\tilde z_j$ leaves the equations $(S\circ M_\phi)\tilde z=0$
unchanged.
Since $\tilde z\geq0$, the three first equations forces the equality of the absolute values of its two terms.
Thus
\[
        \tilde z_1=\tilde z_2=\tilde z_3,
        \qquad
        \tilde z_4=5\tilde z_1.
\]
If $\tilde z_1=0$, then $\tilde z_2=\tilde z_3=\tilde z_4=0$, and the $i$th equation gives
$\tilde z_{4+i}=0$ for all $i$, contradicting $\tilde z\neq0$.
Hence we may assume $\tilde z_1=1$ by rescaling,
and therefore $\tilde z_2=\tilde z_3=1$, $\tilde z_4=5$.
The $(3+i)$th row  of $S\circ M_\phi$ gives
\[
S(3+i,2)+2S(3+i,3)+S(3+i,4+i)\tilde z_{4+i}=0.
\]
It means that $\tilde z_{4+i} = |S(3+i,2)+2S(3+i,3)|\in\{1,3\}$ since $S(3+i,2)+2S(3+i,3)\in\{\pm1,\pm3\}$.
We set
\[
x_i=\frac{\tilde z_{4+i}-1}{2}\in\{0,1\}.
\] 
For the $\ell$th clause in $\phi$ of the form $\textsc{NAE}(x_{i_\ell},x_{j_\ell},x_{k_\ell})$, the corresponding $(n+3+\ell)$th  row of $S \circ M_{\phi}$ has the form $\pm 5 \pm \tilde z_{4+i_\ell} \pm \tilde z_{4+j_\ell} \pm \tilde z_{4+k_\ell}$.  
Since, $\tilde z_{4+i} \in \{1,3\}$ for all $i \in [n]$, if the three entries $z_{4+i_\ell}, z_{4+j_\ell}, z_{4+k_\ell}$ are equal to $1$, the sum can be at most $3$, so it cannot cancel out $5$. Hence, all three  $x_{i_\ell},x_{j_\ell},x_{k_\ell}$ cannot be $0$ together. 
If all these three entries of $\tilde z$ are set to $3$, which corresponds to the case when all  $x_{i_\ell},x_{j_\ell},x_{k_\ell}$ are equal to $1$, the sum $\pm \tilde z_{4+i_\ell} \pm \tilde z_{4+j_\ell} \pm \tilde z_{4+k_\ell} \in \{\pm 3, \pm 9\}$, and hence cannot cancel out $5$. {Instead, if $z_{4+i_\ell},z_{4+j_\ell},z_{4+k_\ell}$ match the configuration $\{1,1,3\}$ or $\{1,3,3\}$, then we can make the corresponding row vanish, since $-5-1+3+3=0=5-1-1-3$.}
Thus we can conclude that $x$ is an assignment that satisfies $\phi$. 
    \end{proof}
\end{theorem}

\paragraph{Discussion.} It follows from the reduction in \eqref{eq:defninstanceAphi} that \textsc{LRMS} over matrices with entries in $\{0,1,2,5\}$ is strongly \textsc{NP}-hard. This  implies that $\textsc{ExactEPMF}$ over matrices with entries in $\{0,1,2^p,5^p\}$ is strongly \textsc{NP}-hard as well. One could ask whether such a hardness result holds over binary matrices; we leave this for a future investigation.

\begin{corollary} \label{thm:exactEPMF-strongly-nphard}
\textsc{ExactEPMF} is strongly NP-hard when $r$ is part of the input.
\begin{proof}
We give a polynomial-time many-one reduction from \textsc{Monotone NAE-3SAT} to \textsc{ExactEPMF}. Let
\[
\phi=\bigwedge_{l=1}^{m}\operatorname{NAE}(x_{i_l},x_{j_l},x_{k_l})\]
be a  \textsc{Monotone NAE-3SAT} instance over the variables $(x_1,\dots,x_n)$. Let  $M_\phi\in~\mathbb{Z}^{(m+n+3)\times(n+4)}$ be as in \eqref{eq:defninstanceAphi}, $X_\phi:=M_\phi^{\circ p}$ and set $r_\phi:=n+3$. Since $p$ is fixed, this construction is computable in polynomial time in the size of $\phi$ and hence $(X_\phi,r_\phi)$ is a valid instance of \textsc{ExactEPMF}.
We show that $(X_\phi,r_\phi)$ is a yes-instance of \textsc{ExactEPMF} if and only if there exists $S \in \{\pm 1\}^{(m+n+3) \times (n+4)}$ such that $\rank(S \circ M_\phi) \leq r_\phi$. 

\textbf{If part.} Suppose that $(X_\phi,r_\phi)$ is a yes-instance of \textsc{ExactEPMF}. Then there exist matrices $W \in \R^{(m + n + 3) \times r_\phi}$ and $H \in \R^{r_\phi \times (n+4)}$ such that $X_\phi = |WH|^{\circ p}$. Then $M_\phi = X_\phi^{\circ \frac{1}{p}} = |WH|$ and hence, there exists a sign matrix $S$ such that $S \circ M_\phi = WH$. Hence $\rank(S\circ M_\phi) \leq r_\phi$. 

\textbf{Only if part.} Suppose there exists a sign matrix $S$ such that  $\rank(S \circ M_\phi) \leq r_\phi$. Then $S \circ M_{\phi} = WH$ for some matrices $W \in \R^{(m + n + 3)  \times r_\phi}$ and $H \in \R^{r_\phi \times (n+4)}$.
Since $M_\phi \geq 0$, taking entrywise absolute values on both sides, we get that $M_\phi = |WH|$ and consequently, $X_\phi = |WH|^{\circ p}$ which shows that $(X_\phi,r_\phi)$ is a yes-instance of \textsc{ExactEPMF}.
Using this and Theorem \ref{thm:satiffsignrank}, we conclude that $\phi$ is a satisfiable instance of \textsc{Monotone NAE-3SAT} if and only if  $(X_\phi,r_\phi)$ is a yes-instance of \textsc{ExactEPMF}.

Since $M_\phi$ has entries in $\{0,1,2,5\}$ and $p\geq 1$ is fixed, the matrix $X_\phi = M_\phi^{\circ p}$ has entries in the set $\{0,1,2^p,5^p\}$.
Hence the reduction uses bounded integer entries, which proves strong NP-hardness.
\end{proof}

\end{corollary}


\section{Tractable algorithms for LRMS and implications for ExactEPMF}\label{sec:exact-fixed-rank}

We now turn to positive results. 
We first show in Section~\ref{sec:polytimefixedr} that LRMS can be solved in polynomial time when $r$ is fixed. In fact, we propose and implement an algorithm to do so (Algorithm~\ref{alg:improvedfinal}). 
Then in Section~\ref{sec:improvFPT}, we propose two improvements of this algorithm which allow us  to show that our algorithm is FPT in the rank $r$ for generic matrices. We illustrate these results numerically on some examples. 
Finally, we show what these complexity results for LRMS extend to  \textsc{ExactEPMF} in Section~\ref{sec:complexexactEPMF}, as both problems are equivalent.  

\subsection{Low-rank matrix signing with fixed rank} \label{sec:polytimefixedr} 

As we have seen in the previous section, exact EPMF is equivalent to LRMS. 
Interestingly, LRMS is complementary to a well-studied problem:  given a sign pattern $S \in \{\pm 1\}^{m \times n}$, find a matrix $X_r$ of rank at most $r$ such that the sign of the entries of $X_r$ coincide with $S$. 
Bhangale and Kopparty~\cite{BK15} 
proved that 
this problem is $\exists\mathbb{R}$-complete already when $r=3$. 
In this section, we prove that,
for every fixed~$r$, LRMS    
is polynomial-time solvable. 

Before doing so, let us recall a result about the skeleton decomposition which we will use to construct a rank-$r$ factorization.   
\begin{lemma} \label{lem:skeleton}
    The matrix $Y\in \R^{m\times n}$ is a rank-$r$ matrix if and only if there exist index sets $I\subseteq [m]$ and $J\subseteq [n]$ such that $|I|=|J|=r$ and $\rank(Y(I,J)) = r$. Moreover, 
    \begin{equation}\label{eq:skeleton}
        Y = Y(:,J) Y(I,J)^{-1} Y(I,:).
    \end{equation}
    \begin{proof}
        The first part of the statement is a well-known fact in linear algebra. For the second part, since $Y(:,J)$ has rank $r$ it is a basis of the image of $Y$; hence there exists $B\in \R^{r\times n}$ such that $Y=Y(:,J)B$. Taking the rows related to the indices of $I$ yields $Y(I,:)=Y(I,J)B$, that is $B=Y(I,J)^{-1} Y(I,:)$, and the claim is straightforward.
    \end{proof}
\end{lemma} 
The skeleton decomposition~\cite{goreinov1997theory} is a special type of CUR approximation;  
see, e.g.,~\cite{kishore2017literature} and the references therein. 

\begin{theorem}[LRMS with fixed rank]\label{thm:fixed-rank-signing}
Let $M\in\mathbb{R}_{\geq 0}^{m\times n}$ and let $r\geq 1$ be a fixed integer.
There is a deterministic algorithm running in time polynomial in $m$ and $n$ that decides whether there exists a sign matrix $S\in\{\pm1\}^{m\times n}$ with $\operatorname{rank}(S\circ M)\leq r$
and outputs the matrix $S$ if it exists.     
\end{theorem}
\begin{proof} 
For each $k=1,\ldots,r$, we need to check whether there exists a matrix of the form $Y = S \circ M$ with rank exactly $k$, where $S\in\{\pm1\}^{m\times n}$. 
If such a rank-$k$ matrix $Y$ exists, then it contains a non-singular $k\times k$ submatrix, say $Y(I,J)$, with $|I|=|J|=k$. Moreover, there exists a rank-$k$ factorization $(W,H)$ such that $Y= WH$ where $W = Y(:,J)$ and $H = Y(I,J)^{-1} Y(I,:)$; see Lemma~\ref{lem:skeleton}. 
The difficulty, of course, lies in choosing the right sign pattern $S$ to obtain $Y$. 
We will fix the signs of $Y(I,J)$ and show that, by using the linear dependencies, we do not need to enumerate all the signs of the other entries of $Y$, keeping our enumeration 
 polynomial in $m$ and $n$. Note that a brute-force enumeration would require $2^{mn}$ operations. The steps below describe our polynomial enumeration of sign patterns $S$ and how to construct a feasible solution $(W,H)$. It may be viewed as the exploration of a branching tree: each branch corresponds to a sequence of
choices for the signs of the entries of $Y := S \circ M$, and if a feasible solution $Y$ exists, then at least one branch will lead to the solution $Y$, otherwise no branch will lead to a feasible solution and the algorithm certifies that the problem is infeasible. 
The pseudocode is provided in Algorithm~\ref{alg:improvedfinal} which also includes several improvements described in the next section.

\begin{algorithm}[ht!] 
\caption{Solving low-rank matrix signing  (Theorems~\ref{thm:fixed-rank-signing} and \ref{thm:FPTgeneric})}
\label{alg:improvedfinal}
\begin{algorithmic}[1]

\Require A non-zero nonnegative matrix $M \in \mathbb{R}^{m \times n}_{\geq 0}$, and a rank $r$.

\Ensure If they exist, rank-$k$ matrices $W \in \mathbb{R}^{m \times k}$ and $H \in \mathbb{R}^{k \times n}$ for the smallest possible $k \leq r$ such that $M = |WH|$. 
Hence $\rank(S \circ M) = k \leq r$ for $S = \sign(WH)$. Otherwise, $W = [\,]$, $H = [\,]$.

\For{$k = 1,\dots,r$}  \label{StepK}

    \ForAll{$I \subset [m]$ and $J \subset [n]$ of size $k$} \Comment{Step 1}

        \State check\_full\_rank\_1 $\gets 1$ 
        \Comment{Improvement Step~1: check Condition~\eqref{eq:condStep1}}  \label{step:improvstep11}

        \ForAll{sign patterns $S(I,J) =: \tilde S \in \{\pm 1\}^{k \times k}$ with $\tilde S(1,:) = 1$, $\tilde S(:,1) = 1$ \label{setp:signpattern}}
        \label{setpsignS}

            \If{$\rank(\tilde S \circ M(I,J)) < k$}

                \State check\_full\_rank\_1 $\gets 0$ \Comment{Condition~\eqref{eq:condStep1} fails} \label{step:improvstep12}  

                \State Go to line~\ref{setp:signpattern} and try the next sign pattern. 

            \Else

                \State Construct the sets $R_i$, $i\in[m]$, and $C_j$, $j\in[n]$, as in~\eqref{eq:setsRi} and~\eqref{eq:setsCi}.

                \ForAll{subsets $P$ of $[m]$ of size $\min\big(m,\binom{k+1}{2}\big)$}
                \label{setpP}  \Comment{Steps 2 \& 3} 

                    \State check\_full\_rank\_2 $\gets 1$ \Comment{Improvement Step~3: check Condition~\eqref{eq:condStep3}}

                    \ForAll{$\{w_i \}_{i \in P}$ with $w_i \in R_i$} \Comment{$R_i$ defined in \eqref{eq:setsRi}} \label{setpChoiceWi}

                        \If{$\{w_i w_i^\top\}_{i \in P}$ is not a basis of $k\times k$ symmetric matrices}

                            \State check\_full\_rank\_2 $\gets 0$ \Comment{Condition~\eqref{eq:condStep3} fails}

                        \EndIf

                        \For{$j = 1,\dots,n$} \Comment{Step 4} 

                            \State Find any $h_j \in C_j$ from~\eqref{eq:setsCi} with $|w_i^\top h_j| = M(i,j)$ for all $i \in P$. 

                            \State If not possible,
                            go to line~\ref{setpChoiceWi} and try the next choice of $\{w_i \}_{i \in P}$.

                        \EndFor

                        \ForAll{$i \notin I\cup P$} \Comment{Step 5} 

                            \State Find any $w_i \in R_i$ such that $|w_i^\top h_j| = M(i,j)$ for all $j$.

                            \State If not possible, 
                            go to line~\ref{setpChoiceWi} and try the next choice of $\{w_i \}_{i \in P}$.
                             
                        \EndFor

                        \State Solution found:  \Return $[W,H]$.

                    \EndFor

                    \If{check\_full\_rank\_2 $= 1$} \Comment{Condition~\eqref{eq:condStep3} satisfied but no solution} 

                        \State Go to line~\ref{setp:signpattern} and try the next sign pattern. 

                    \EndIf

                \EndFor

            \EndIf

        \EndFor

        \If{check\_full\_rank\_1 $= 1$} \Comment{Condition~\eqref{eq:condStep1} satisfied but no solution} 

            \State Go to line~\ref{StepK} and try the next $k$. \label{step:improvstep13}  
        \EndIf

    \EndFor

\EndFor

\State \Return $W=[\,],H=[\,]$. 

\end{algorithmic}
\end{algorithm}

\medskip
\noindent\textbf{Step~1: Choosing a nonsingular basis block.}
Enumerate all pairs $(I,J)$ with $I\subseteq[m]$, $J\subseteq[n]$,
$|I|=|J|=k$, and all sign patterns
$S(I,J) \in\{\pm1\}^{k\times k}$ for $k \in [r]$.
Set $Y(I,J) = S(I,J) \circ M(I,J)$, but discard the branch if $Y(I,J)$ is singular. In summary, after this step, we have fixed the sign pattern of a $k\times k$ submatrix of $S$ indexed by $(I,J)$ such that $Y(I,J) = S(I,J) \circ M(I,J)$ is non singular. The number of branches at Step~1 is at most
$\binom{m}{k}\binom{n}{k}2^{k^2}\leq \mathcal{O}\bigl((mn)^r\,2^{r^2}\bigr)$. 

\medskip
\noindent\textbf{Step~2: Column and row candidates.}  
For a branch with a $k\times k$ nonsingular submatrix $Y(I,J)$ indexed by $(I,J)$, any compatible rank-$k$ matrix may use the columns indexed by $J$ as a basis of its column space, that is, $W = Y(:,J)$. 
Let us denote $w_i^\top \in \mathbb{R}^k$ the $i$th row of $W \in \mathbb{R}^{m \times k}$ for $i\in[m]$. 
We know $W(I,:) = Y(I,J)$  
since these signs have been fixed.   
For the other rows of $W$, we can still pick any signs for $Y(i,J)$ for $i \notin I$, with a total of $2^k$ possibilities for each row. In summary, for all $i\in[m]$, 
\begin{equation} \label{eq:setsRi}
w_i  \; \in  \;  R_i \;=\;
  \begin{cases}
    \bigl\{u\in\mathbb{R}^{k}:|u|=M(i,J)^\top \bigr\}  
      & \text{ for }i\notin I,\\[4pt]
    \{ Y(i,J)^\top \} &  
    \text{ for } i \in I. 
  \end{cases} 
\end{equation}
Similarly, for the $j$th column of $H \in \mathbb{R}^{k \times n}$, denoted $h_j$, 
we have, using the skeleton decomposition, 
\begin{equation} \label{eq:setsCi}
h_j \; \in   \; 
  C_j \; = \;
  \begin{cases}
    \bigl\{Y(I,J)^{-1} y : y\in\mathbb{R}^{k},\;|y|=M(I,j)\bigr\}  
      & \text{ for } j\notin J,\\[4pt]
    \{e_{\ell}\} & \text{ for } j \in J \text{ and $j$ is the $\ell$th entry in $J$}. 
  \end{cases}
\end{equation}
Note that $H(:,J) = I_k$, the $k\times k$ identity matrix. 

Suppose that a feasible rank-$k$ matrix $Y = S \circ M$ exists, where $S(I,J)$ was fixed (Step~1). By the skeleton decomposition, there must exist vectors $w_i\in R_i$ and $h_j \in C_j$ such that $Y(i,j) = w_i^\top h_j$ for all $i,j$, hence $|Y(i,j)| = M(i,j) = |w_i^\top h_j|$ for all $i,j$. 
Finding the candidates $w_i\in R_i$ and $h_j \in C_j$ such that $M(i,j) = |w_i^\top h_j|$ for all $(i,j)$ hence provides a certificate that a feasible solution is found (if such candidates do not exist, then one cannot play with the signs of $M$ to make it rank $k$). 
Unfortunately, checking all candidates does not lead to the desired results (namely, having an algorithm polynomial in $m$ and $n$), as for each of the remaining $m-k$ rows (resp.\ $n-k$ columns) of $S$, there are $(2^k)^{m-k}$ (resp.\ $(2^k)^{n-k}$) candidates. 
The next step will select a sufficiently small subset of the row candidates to achieve 
our goal.  

\medskip
\noindent\textbf{Step~3: Selecting row candidates.}  
For any feasible choice of row and column candidates $w_i \in R_i$ and $h_j \in C_j$ for all $(i,j)$, the feasibility condition $|w_i^\top h_j|=M(i,j)$ between row $i$ and column $j$ can be rewritten as 
\begin{equation}\label{eq:linearxxuu}
\langle h_j h_j^\top,\,w_i w_i^\top\rangle = M(i,j)^2.
\end{equation}
The matrices $w_iw_i^\top$ lie in the space of $k\times k$ symmetric matrices, whose dimension is $\binom{k+1}{2}$.
Hence, there exists
$P\subseteq[m]$ with $|P|\leq\binom{k+1}{2}$ such that
$w_\ell w_\ell^\top$  lies in the linear span of
$\{w_iw_i^\top:i\in P\}$ for all  $\ell\in[m]$. 
Note that if $m \leq \binom{k+1}{2}$, we can simply pick all rows and then keep only the corresponding linearly independent matrices. 
This linear dependence will allow us to filter the column candidates using only the rows indexed by $P$, and then to complete the remaining rows independently. 
In order to identify if a matrix $Y$ exists, we branch over all possible subsets~$P$ of size $\min\big(m, \binom{k+1}{2}\big)$ 
and all possible choices
of $w_i\in R_i$ for $i\in P$, creating at most
\[
  \binom{m}{\binom{k+1}{2}}(2^k)^{\binom{k+1}{2}}
  = m^{\cO(k^2)}\,2^{\cO(k^3)}
  \leq m^{\cO(r^2)}\,2^{\cO(r^3)}
\]
different new branches to explore. 
Note that we do not check whether $\{w_iw_i^\top\}_{i \in P}$ spans $\{w_iw_i^\top\}_{i \in [m]}$ for all possible $w_i \in R_i$ since this would be too expensive: what is important is that at least one such $P$ will be such that $\{w_iw_i^\top\}_{i \in P}$ spans $\{w_iw_i^\top\}_{i \in [m]}$. However, the algorithm might find a solution for a $P$ that does not satisfy this condition. 

\medskip
\noindent\textbf{Step~4: Filtering column candidates.} 
In step~3, we identified the candidates $w_i\in R_i$ for $i\in P$, each of which we can explore as a new branch (since there are sufficiently few, that is, a polynomial number in $m$). 
In other words, the subset of rows of $W$ in $P$ of size at most $\binom{k+1}{2}$ is fixed and the candidate $w_i$ is selected for all $i \in P$. 
We use them to check which column candidates for $H$ remain compatible with the candidate rows for $W$. 
For each $j\in[n]$, we retain only the candidates $h_j \in C_j$ satisfying
$(w_i^\top h_j)^2=M(i,j)^2$ for all $i\in P$.  
We discard the branch if no candidate $h_j \in C_j$ satisfies these conditions for some $j$, 
otherwise we select arbitrarily one surviving candidate per column.
The justification for keeping an arbitrary survivor is given in the next step.

\medskip
\noindent\textbf{Step~5: Selecting the remaining row candidates.} 
At the end of Step~4, one candidate $h_j\in C_j$ has been selected for every column $j\in[n]$.
It remains to select the row candidates $w_i\in R_i$ for $i\notin I\cup P$.
By construction, the selected $h_j$'s satisfy 
\[
\langle h_jh_j^\top,w_i w_i^\top\rangle=M(i,j)^2 \qquad\text{for all }i\in P.
\]
For each $i\notin I\cup P$, we select the vectors $w_i \in R_i$ satisfying  $$(w_i^\top h_j)^2=M(i,j)^2 \qquad\text{for every }j\in[n].$$
If no candidate survives for some $i\notin I\cup P$, we discard the branch.
Otherwise, we select arbitrarily one surviving candidate $w_i$. 

Let us show that the arbitrary choice of the $h_j$'s made in Step~4 are safe in a branch that leads to a feasible solution $Y$. 
Let $h_j^\ast$ be another admissible choice for column $j$ in Step~4, that is, this column coordinate $h_j^\ast$ satisfies  
\[
\langle h_j^\ast (h_j^\ast)^\top,w_i w_i^\top\rangle  = 
 M(i,j)^2 = \langle h_jh_j^\top,w_i w_i^\top\rangle \qquad\text{for all }i\in P.
\]
For $\ell\in[m]$, if $\{w_i w_i^\top\}_{i \in P}$ spans $\{w_i w_i^\top\}_{i \in [m]}$, there exist scalars 
$\alpha_i$ such that
$w_\ell w_\ell^\top
=\sum_{i\in P}\alpha_i\,w_i w_i^\top$.
Therefore,
\[
\begin{aligned}
\langle h_jh_j^\top,w_\ell w_\ell^\top\rangle  = \sum_{i\in P}\alpha_i\langle h_jh_j^\top,w_iw_i^\top\rangle 
 = \sum_{i\in P}\alpha_i M(i,j)^2  
 = \sum_{i\in P}\alpha_i\langle h_j^\ast (h_j^\ast)^\top,w_i w_i^\top\rangle 
 & = \langle h_j^\ast (h_j^\ast)^\top,w_\ell w_\ell^\top\rangle  \\ 
 & = M(\ell,j)^2.
\end{aligned}
\]
This means that the choice of any admissible $h_j$ in Step~4 does not affect the value of $\langle h_jh_j^\top, w_\ell w_\ell^\top\rangle$ for all $\ell \in [m]$, because of the linear dependence.  

\medskip
\noindent\textbf{Step~6: Constructing the factorization.} 
If there were no rank-$k$ feasible solution $Y = S \circ M$, then no branch generated in Steps~1-5 will remain, and we can conclude that the rank-$k$ problem is not feasible. 
Otherwise, our procedure provides a rank-$k$ factorization $(W,H)$ such that $M = |WH|$. 
In fact, let $W$ and $H$ contain the selected row and column candidates, respectively, that is, 
\[
W=\begin{pmatrix}w_1^\top\\\vdots\\w_m^\top\end{pmatrix}
  \in\mathbb{R}^{m\times k},
  \qquad
  H=
\begin{pmatrix}h_1&\cdots&h_n\end{pmatrix}
  \in\mathbb{R}^{k\times n}, \quad\text{ and }\quad S=\operatorname{sign}(WH)\in\{\pm 1\}^{m\times n}.
\]
By the filtering performed in the previous steps, the selected candidates satisfy $|W(i,:)\,H(:,j)|=|w_i^\top h_j|=M(i,j)$ for every
$i\in[m]$ and $j\in[n]$, so $\operatorname{rank}(S\circ M)=
\operatorname{rank}(WH)\le k\leq r$. 

\medskip
\noindent\textbf{Running time.}
For each $k\le r$, Step~1 produces $\cO\bigl((mn)^k\,2^{k^2}\bigr)$ branches and Step~3 produces $m^{\cO(k^2)}\,2^{\cO(k^3)}$ further branches.  Steps~4--5
cost $\operatorname{poly}(m,n)$ per branch.  The total is\footnote{Moreover, for rational inputs, if the bit-size of the entries of $M$ is at most $b$, for $r$ fixed, every rational number occurring in the computation has bit-size at most $\text{poly}_r(b)$ and so the algorithm is strongly polynomial over $\mathbb{Q}$.} 
$(mn)^r\cdot m^{\cO(r^2)}\cdot 2^{\cO(r^3)}\cdot\operatorname{poly}(m,n,r)$.
This concludes the proof. 

Note that the term $m^{\cO(r^2)}$ comes from the fact the we first explore the possible signs of the rows of $W$ in the sets $R_i$'s. If $n < m$, it is preferable to apply the algorithm on the transpose to have complexity $n^{\cO(r^2)}$. 
\end{proof}

\subsection{Improvements of the algorithm and FPT for generic matrices}   \label{sec:improvFPT}

In this section, we first provide two general improvements of the algorithm described in Theorem~\ref{thm:fixed-rank-signing}, one for Step~1 in Sections~\ref{sec:improvstep1} and one for Step~3 in Section~\ref{sec:improvstep3}. 
These allow us to provide, for LRMS, a fixed-parameter tractable (FPT) algorithm in the parameter $r$ for rank-$\rb$ generic matrices, where $\rb$ is a fixed positive integer, with probability one (w.p.1) in  Section~\ref{sec:FPTgeneric}.

\subsubsection{Improvements of Step~1} 
\label{sec:improvstep1}

We have the following result. 

\begin{lemma}
    Given $1\leq k\leq r$, let us consider Step~1 of the algorithm described in  Theorem~\ref{thm:fixed-rank-signing}. Assume we have found a pair of subsets $(I,J)$ with $|I|=|J|=k$ such that 
\begin{equation} \label{eq:condStep1} 
\rank\Big( M(I,J) \circ S(I,J) \Big) = k \quad  \text{ for all } \quad  S(I,J) \in \{\pm 1\}^{k \times k}. 
\end{equation} 
Then one does not need to consider other pairs of subsets $(I,J)$ to guarantee whether there exists $S$ such that $\rank(M \circ S) \leq k$. 
\end{lemma}
\begin{proof}
The condition~\eqref{eq:condStep1} means that, regardless of its sign pattern, the submatrix $M(I,J)$ always has rank $k$, that is, we cannot reduce its rank by playing with its signs. This implies that, regardless of the sign pattern $S$, we can always construct a skeleton decomposition using the indices $(I,J)$ since $S(I,J) \circ M(I,J)$ is invertible for all sign patterns $S(I,J)$, and hence we do not need to explore other subsets of indices in the algorithm described in Theorem~\ref{thm:fixed-rank-signing} that construct a skeleton decomposition of $Y = M \circ S$.  
\end{proof}

As we will see in Section~\ref{sec:FPTgeneric}, condition~\eqref{eq:condStep1} is valid w.p.1 for generic matrices, and hence only one subset of indices $(I,J)$ needs to be checked in that case, leading to an FPT complexity in the parameter $r$ for Step~1. 

Another improvement for Step~1 (line \ref{setp:signpattern} in Algorithm \ref{alg:improvedfinal}) is that not all sign patterns need to be checked for $M(I,J)$. 
In fact, the rank of a matrix is unchanged under sign flips of its rows and columns, hence we can assume w.l.o.g.\ that the first row and first column of $S(I,J)$ have a fixed sign pattern, e.g., $\tilde S(1,:) = \tilde S (:,1) = 1$, where $\tilde S := S(I,J)$. 
This means that we only need to check $2^{(k-1)^2}$ sign patterns instead of $2^{k^2}$. 

\subsubsection{Improvement of Step~3} \label{sec:improvstep3} 

The other step that is not FPT in the parameter $r$ in Theorem~\ref{thm:fixed-rank-signing} is Step~3. Recall that, in Step~3, we need to check all subsets $P$ of size  $\min\big(m,\binom{k+1}{2}\big)$ of rows of $W$, for $1\leq k \leq r$. In a similar spirit as for Step~1, if for a subset $P$, $\{w_i w_i^\top\}_{i \in P}$ is a basis of the symmetric matrices for any possible choice of the signs (that is, for any $w_i \in R_i$), we do not need to consider other subsets. 
\begin{lemma}
    Given $1\leq k\leq r$, let us consider Step~3 of the algorithm of Theorem~\ref{thm:fixed-rank-signing}. 
    Assume that the subset $P$ of size $\min\big(m,\frac{k(k+1)}{2}\big)$ is such that  
    \begin{equation} \label{eq:condStep3}
         \{w_iw_i^\top\}_{i \in P} \text{ are linearly independent for all } w_i \in R_i.           
    \end{equation}   
Then, in the algorithm of Theorem~\ref{thm:fixed-rank-signing}, one does not need to consider other subsets $P$ in  
Step~3 for the current sign pattern $S(I,J)$.   
\end{lemma}
\begin{proof} Since $\{w_iw_i^\top\}_{i \in P}$ are linearly independent regardless of the choice of the $w_i$'s in $R_i$, all other rank-one matrices, $\{w_iw_i^\top\}_{i \notin P}$, will be in the span of $\{w_iw_i^\top\}_{i \in P}$. Hence, when looking for $h_j \in C_j$ in Step~5, the argument described in Step~5 allows us to only check the conditions  
$\langle h_jh_j^\top,w_i w_i^\top\rangle=M(i,j)^2$ for $i\in P$, as it will imply $\langle h_jh_j^\top,w_i w_i^\top\rangle=M(i,j)^2$ for $i\notin P$. 
\end{proof}

As we will see in Section~\ref{sec:FPTgeneric}, condition~\eqref{eq:condStep3} is valid w.p.1 for generic matrices, and hence only one set $P$ needs to be checked in that case, leading to an FPT complexity in the parameter $r$ for Step~3.

Algorithm~\ref{alg:improvedfinal} provides the pseudo-code of the algorithm described in Theorem~\ref{thm:fixed-rank-signing}, combined with the improvements in Steps~1 and~3 described in the two previous sections, which allows to stop the $k$th step early when \eqref{eq:condStep1} is satisfied, or avoid enumerating all subsets $P$ for a given sign pattern $S(I,J)$ when \eqref{eq:condStep3} is satisfied.

\subsubsection{Fixed-parameter tractable algorithm in the parameter \texorpdfstring{$r$}{r} for generic matrices} \label{sec:FPTgeneric}

For simplicity, we consider generic matrices of the following form: the matrix $Y$ is a generic rank-$\rb$  matrix if $Y = W_{\#} H_{\#}$, where the entries of $W_{\#} \in \mathbb{R}^{m \times \rb}$ and $H_{\#} \in \mathbb{R}^{\rb \times n}$ are drawn from a continuous distribution (e.g., Gaussian or uniform). 

Before proving that Algorithm~\ref{alg:improvedfinal} is fixed-parameter tractable (FPT) in the parameter $r$ for such generic rank-$\rb$  matrices (Theorem~\ref{thm:FPTgeneric}), we provide several lemmas. The first lemma shows that for a generic rank-$\rb$ matrix $Y$, the $k\times k$ submatrix $Y(I,J) \circ \tilde S$ is full rank w.p.1 for any sign pattern $\tilde S$ and any $1 \leq k \leq \rb$.    
\begin{lemma} \label{lem:genericlowranksign}
    \label{lem:generic_full_rank}
    Let $Y=W_{\#}H_{\#}$, where the entries of 
    $W_{\#} \in \R^{m\times \rb}$ and 
    $H_{\#} \in \mathbb{R}^{\rb \times n}$ are drawn from a continuous distribution. Then 
    \begin{equation}
        \label{eq:sumprobzero}
        \sum_{k=1}^\rb \sum_{\substack{S\in \{\pm 1\}^{m\times n} \\ I \subset [m],J\subset [n]: |I|=|J|=k}}
        \p\left(
        \det\big( 
        S(I,J)\circ Y(I,J)
        \big)=0
        \right) \quad = \quad 0,
    \end{equation}
    which means that for all $k$, for every choice of the matrix $S$ and of the sets $I$ and $J$ of size $k$, the submatrix $S(I,J)\circ Y(I,J)$ has full rank w.p.1. 
    \end{lemma} 
    \begin{proof}
       Since the sum in \eqref{eq:sumprobzero} runs on a finite set, it is enough to show the claim for $k$, $S$, $I$ and $J$ fixed. The polynomial
       \[
        {f}(W(I,:),H(:,J)) = \det(S(I,J)\circ (W(I,:)H(:,J)))
       \]
       is non-zero since, for the choice $W(I,:)$ and $H(:,J)=$ so that $W(I,:)H(:,J)=I_k$ (possible since $k\leq \bar{r}$), it gives $\det(S(I,J)\circ I_k)=\prod_{i=1}^kS(I,J)_{ii}=\pm 1\neq 0$.
       Then the claim follows from the fact that the set of the zeros of ${f}$ has measure zero and that $W_{\#}$ and $H_{\#}$ are picked from a continuous distribution. 
    \end{proof}

The next two lemmas will be used to prove that $\{w_i w_i^\top\}_{i \in P}$ form a basis of the set of symmetric matrices w.p.1 for generic input matrices. 
The first one shows that invertible linear transformations of a continuous distribution remain a continuous distribution. 
\begin{lemma}\label{lem:ctsjtdistbn}
    If a set of vectors $\{w_1,\dots,w_\ell\} \in \R^d$ has an absolutely continuous joint distribution on $(\R^d)^\ell$, then for any $q\leq d$ and for any fixed matrices $B_1, \dots, B_\ell \in \R^{q \times d}$ with full row rank, $\{B_1w_1,\dots, B_\ell w_\ell\}$ also has an absolutely continuous joint distribution on $(\R^q)^\ell$.
    \begin{proof}
        This is a consequence of a much more general statement about composition of Lebesgue measurable functions. For more details and the general statement, refer to \cite[Theorem 2.47]{folland2013real}.
    \end{proof}
\end{lemma}

The second lemma provides the desired result regarding the span of $\{w_i w_i^\top \}_{i \in P}$. 
 
\begin{lemma}\label{lem:genericsymapp}
Let $Y = W_{\#}H_{\#}$, where the entries of $W_{\#} \in \R^{m\times \rb}$ and $H_{\#} \in \mathbb{R}^{\rb \times n}$ are drawn independently from a continuous distribution. The matrix $M = |Y|$ is given as input to  Algorithm~\ref{alg:improvedfinal}. For any $k \leq \rb$, for any subset $P \subseteq [m]$ of size $N_k = \binom{k+1}{2}$, the matrices $\{w_iw_i^\top\}_{i \in P}$ are linearly independent for all $w_i \in R_i$, and hence they span the entire space of $k \times k$ symmetric matrices w.p.1.  

  \begin{proof} 
Let $P $ be any subset of $[m]$ of size $N_k = \binom{k+1}{2}$. For the simplicity of exposition, without loss of generality, we assume $P = \{1,\dots,N_k\}$. 
Define the map
\[
\phi:\mathbb{R}^k \to \mathbb{R}^{N_k},
\qquad
\phi(w)=\operatorname{svec}(ww^\top),
\]
where $\operatorname{svec}$ denotes any vectorization of the independent entries of a symmetric matrix, that is,  the values in the upper triangular part.  
The matrices $\{w_iw_i^\top\}_{i=1}^{N_k}$ span the space of $k\times k$ symmetric matrices if and only if the vectors 
$\phi(w_1),\dots,\phi(w_{N_k})$ 
form a basis of $\mathbb{R}^{N_k}$, that is, if and only if the matrix
\[
A(w_1,\dots,w_n):=
\begin{bmatrix}
\phi(w_1)^\top\\
\vdots\\
\phi(w_{N_k})^\top
\end{bmatrix}
\in \mathbb{R}^{N_k\times N_k}
\]
is nonsingular. Since the entries of $\phi(w)$ are quadratic polynomials in the entries of $w$, the determinant
${f}(w_1,\dots,w_{N_k})=\det(A(w_1,\dots,w_n))$  
is a polynomial in the entries of $w_1,\dots,w_{N_k}$.   
Consider the following collection of $N_k$ vectors:
\[
\mathcal{V} = 
\Big\{\frac{e_i+e_j}{2}, \text{ } 1\le i\leq j\le k\Big\}.
\]
Then the polynomial ${f}$ evaluated on these $N_k$ vectors is non-zero and hence ${f}$ is not the zero polynomial.
Since the entries of $W_{\#}$ are picked independently at random from a continuous distribution, the vectors $\{W_{\#}(i,:)^\top\}_{i \in P}$ have a joint continuous distribution in $(\R^{k})^{N_k}$.  
Following the definition in \eqref{eq:setsRi}, for some $k \leq \rb$, the row candidates $w_i \in R_i$ are such that
 $w_i  = (W_{\#}(i,:)H_{\#}(:,J))^\top \circ s_i$ for some sign vector $s_i \in \{\pm 1\}^k$ where $J \subseteq [n]$, $|J| = k$. For a fixed collection of sign vectors $ \{s_{1},\dots,s_{N_k}\} =: \sigma \in (\{\pm 1\}^k)^{N_k}$, taking $D_i = \diag(s_i)$, we have 
\begin{equation}\label{eq:relwiwiinp}
w_i = D_iH_{\#}(:,J)^\top W_{\#}(i,:)^\top.    
\end{equation}
By Lemma \ref{lem:ctsjtdistbn} for $d = \rb$, $q = k$ and $\ell = N_k$, the corresponding candidate row vectors $w_1,\dots, w_{N_k}$  
have a joint continuous distribution.  Hence, using the fact that the set of zeros of $f$ has measure zero, 
\[
\mathbb{P} \Big(
f(w_1,\dots,w_{N_k}) = 0 
\; \Big| \;  
\rank((D_i(H_{\#}(:,J))^\top) = k \Big) = 0. 
\] 
Since the entries of $H_{\#}(:,J)$ are picked independently at random from a continuous distribution, it holds that $\rank(D_i(H_{\#}(:,J))^\top)=k$ w.p.1, similarly as in Lemma~\ref{lem:genericlowranksign}. Using the fact that there are a finite number of $\sigma$, using union bound, we can therefore conclude the matrices $\{w_iw_i^\top\}_{i \in P}$ are linearly independent and hence, they span the entire space of $k \times k$ symmetric matrices.   
  \end{proof}
\end{lemma}

We can now prove that Algorithm~\ref{alg:improvedfinal} is FPT in the parameter $r$ for generic input matrices. 

\begin{theorem} \label{thm:FPTgeneric}
Given $\rb\geq 1$, let $M = |Y|$ where $Y = W_{\#}H_{\#}$, such that the entries of $W_{\#} \in \R^{m\times \rb}$ and $H_{\#} \in \mathbb{R}^{\rb \times n}$ are drawn from a continuous distribution. Then Algorithm~\ref{alg:improvedfinal} run on $(M,r)$ for any $r$ is FPT in parameter $r$, as it runs in {$2^{\cO(\tilde r^3)}\cdot \cO(mn)$}
operations, where $\tilde r = \min(r, \rb)$.   
\end{theorem} 
\begin{proof} 
If $r \geq \rb$, Algorithm~\ref{alg:improvedfinal} will terminate for $k\leq \rb$ since a feasible solution of rank $\rb$ exists, otherwise the algorithm will stop for $k\leq r\leq \rb$.
In all cases, $k$ will not exceed $\tilde{r}$. 

Now, for any $k \leq \tilde{r}$, the two conditions~\eqref{eq:condStep1} and~\eqref{eq:condStep3}  are satisfied  w.p.1 for any $(I,J)$ and any $P$, as shown in Lemmas~\ref{lem:genericlowranksign} and~\ref{lem:genericsymapp}. 
Hence, Algorithm~\ref{alg:improvedfinal} will only enter once the loop over the indices $(I,J)$ and over the subset $P$, w.p.1. 
In the loop over $P$, it might still need to check all possible $w_i \in R_i$ for $i \in P$, 
with at most $2^{k^2(k+1)/2}$ different choices -- there are $k(k+1)/2$ vectors $w_i \in \mathbb{R}^k$ whose entries can take two signs (unless $i \in I$); see~\eqref{eq:setsRi}.   
Then 
\begin{itemize}
    \item Checking linear independence of $\{w_iw_i^\top\}_{i \in P}$ takes $\cO(k^6)$. 

    \item Finding $h_j \in C_j$ such that $|w_i^\top h_j| = M(i,j)$ takes up to $2^k k \binom{k+1}{2}$ since $C_j$ has up to $2^k$ elements, the inner product costs $k$ operations and there are $\binom{k+1}{2}$ vectors $w_i$. 
    This has to be done $n$ times. 

    \item Finding $w_i \in R_i$ for $i \notin P$ such that $|w_i^\top h_j| = M(i,j)$ for all $j$ takes up to $2^k k n$ which has to be done $m$ times. 
\end{itemize} 
Since when $n>m$ it is always possible to consider $Y^\top$, the total cost is  
\[
\mathcal{O} \left( \tilde r 
2^{\tilde r^3} ( \tilde r^6 + {\tilde r}^3 2^{\tilde r}  \min(m,n) + \tilde r 2^{\tilde r}   m n )
\right). 
\]
\end{proof}

We conclude this section with four remarks. 

\begin{remark}[Fixed-parameter linear] 
For generic matrices, Algorithm~\ref{alg:improvedfinal} is fixed-parameter linear in $mn$, the number of matrix entries, with parameter $r$.  
\end{remark}

\begin{remark}[Generic matrices]
 In this section, we considered rank-$r$ generic matrices as the product of two smaller generic matrices. 
However, one could instead consider the more general case where the matrices are drawn from the manifold of rank-$r$ matrices $\mathcal{M}_r$. Although this space has zero ambient Lebesgue measure in $\mathbb R^{m\times n}$ whenever $r<\min\{m,n\}$, one can alternatively consider the restricted Hausdorff measure and then use that to define the notion of absolutely continuous probability measures on $\mathcal{M}_r$~\cite{Federer_1996}.  The proofs in this section generalize to that setting as well, but we use the current setting for the simplicity of the exposition.  
\end{remark}

\begin{remark}[Last improvement of Algorithm~\ref{alg:improvedfinal}]
A drawback of Algorithm~\ref{alg:improvedfinal} is that it tries to find a rank-$k$ solution, incrementing $k$ from 1 to $r$. Hence if the input matrix has a rank-$r$ solution, but no smaller rank solution, it will first have to enumerate all $k=1,2,\dots,r-1$ before finishing, 
which will be slow. 
To avoid this, one could start the iteration at $k=r$, but then we do not
 necessarily have an FPT algorithm for generic matrices. Hence, in practice, it could actually be better to run two algorithms in parallel: one that increments $k$ from 1, and the other that decrements from $r$. 
\end{remark}

\begin{remark}[Domain of $p$] We could extend our analysis and relax the constraint on $p$ being integer. In particular, whenever the $\big(\frac{1}{p}\big)$th power is an available operation, we can consider $p$ to be a positive real number and, if $X$ has strictly positive entries, we can even allow $p<0$. For the sake of a clearer presentation, we did not consider this more general setting. 
\end{remark}

\subsubsection{Implementation and numerical examples} \label{sec:implementnumexp}

We have implemented Algorithm~\ref{alg:improvedfinal} in MATLAB. It is available from 
\begin{center}
  \url{https://gitlab.com/ngillis/rank-r_signing/}.  
\end{center} 
To give an idea of the size of matrices one can handle on a laptop\footnote{Experiments performed with a 
12th Gen Intel(R) Core(TM) i9-12900H  2.50 GHz, 32GB RAM, \textsc{Matlab} R2019b.}, 
Table~\ref{tab:runtime} reports the computational time needed for the algorithm to generate a feasible solution on rank-$r$ randomly generated matrices. 
When there exists a feasible solution for $r \in \{2,3\}$, our algorithm can handle matrices of large dimensions, up to $10^4$, within seconds. 
For $r=4$, the algorithm starts to stall. In fact, the largest term in the computational cost in Step~3, 
$2^{(r-1)^3}$, goes from 256 for $r=3$ to 
134,217,728 for $r=4$, so this quick change of behavior is expected. 

\begin{table}[ht]
\centering 
\begin{minipage}{0.24\textwidth}
\centering
\begin{tabular}{|c|c|}
\hline
\multicolumn{2}{|c|}{$r=2$} \\ \hline
$m=n$ & time (s.) \\ \hline
100  & $0.005 \pm 0.008$ \\ 
500  & $0.02 \pm 0.03$ \\ 
1000 & $0.07 \pm 0.04$ \\ 
10000 & $8.57 \pm 0.1$ \\ 
\hline
\end{tabular}
\end{minipage}
\hfill
\begin{minipage}{0.24\textwidth}
\centering
\begin{tabular}{|c|c|}
\hline
\multicolumn{2}{|c|}{$r=3$} \\ \hline
$m=n$ & time (s.) \\ \hline
100  & $0.03 \pm 0.02$ \\ 
500  & $0.08 \pm 0.03$ \\ 
1000 & $0.18 \pm 0.08$ \\ 
10000 & $9.94 \pm 0.6$ \\
\hline
\end{tabular}
\end{minipage}
\hfill
\begin{minipage}{0.22\textwidth} 
\centering
\begin{tabular}{|c|c|}
\hline
\multicolumn{2}{|c|}{$r=4$} \\ \hline
$m=n$ & time (s.) \\ \hline
5 & $0.04 \pm 0.03$ \\ 
7 & $1.80 \pm 1.06$ \\ 
9 & $118 \pm 68.5$ \\ 
10 & $> 600$ \\ 
\hline
\end{tabular}
\end{minipage}
\hfill
\begin{minipage}{0.22\textwidth} 
\centering
\begin{tabular}{|c|c|}
\hline
\multicolumn{2}{|c|}{$r=5$} \\ \hline
$m=n$ & time (s.) \\ \hline
6 & $11.4  \pm   6.4$ \\ 
7 & $178 \pm    104$ \\ 
8 & $> 600$ \\ 
\hline
\end{tabular}
\end{minipage}
\caption{Average run time in seconds, and standard deviation, to certify feasibility among 20 runs on rank-$r$ randomly generated matrices
$M = |\texttt{randn}(m,r)*\texttt{randn}(r,n)|$. \label{tab:runtime}}
\end{table}

Table~\ref{tab:runtime2} reports experiments for full-rank generic matrices, to see how long the algorithm takes to certify that no rank-$k$ solution exist for $k \leq r$. 
\begin{table}[ht]
\centering 
\begin{minipage}{0.24\textwidth}
\centering
\begin{tabular}{|c|c|}
\hline
\multicolumn{2}{|c|}{$r=2$} \\ \hline
$m=n$ & time (s.) \\ \hline
100  & $0.002 \pm 0.002$ \\ 
500  & $0.007 \pm 0.0004$ \\ 
1000 & $0.014 \pm 0.0007$ \\ 
10000 & $0.14 \pm 0.005$ \\ 
\hline
\end{tabular}
\end{minipage}
\hspace{0.3cm} 
\begin{minipage}{0.24\textwidth}
\centering
\begin{tabular}{|c|c|}
\hline
\multicolumn{2}{|c|}{$r=3$} \\ \hline
$m=n$ & time (s.) \\ \hline
100  & $0.04 \pm 0.001$ \\ 
500  & $0.13 \pm 0.003$ \\ 
1000 & $0.24 \pm 0.008$ \\ 
10000 & $2.21 \pm 0.05$ \\
\hline
\end{tabular}
\end{minipage}
\hspace{0.05cm}  
\begin{minipage}{0.22\textwidth} 
\centering
\begin{tabular}{|c|c|}
\hline
\multicolumn{2}{|c|}{$r=4$} \\ \hline
$m=n$ & time (s.) \\ \hline
5 & $0.07 \pm 0.003$ \\ 
7 & $3.11 \pm 0.02$ \\ 
9 & $217 \pm 0.44$ \\ 
10 & $> 600$ \\ 
\hline
\end{tabular}
\end{minipage}
\hspace{0.2cm}  
\begin{minipage}{0.22\textwidth} 
\centering
\begin{tabular}{|c|c|}
\hline
\multicolumn{2}{|c|}{$r=5$} \\ \hline
$m=n$ & time (s.) \\ \hline
6 & $21.1  \pm  0.08 $ \\ 
7 & $286  \pm  0.72$ \\ 
8 & $> 600$ \\ 
\hline
\end{tabular}
\end{minipage}
\caption{Average run time in seconds, and standard deviation, to certify infeasibility among 20 runs on full-rank randomly generated matrices
$M = |\texttt{randn}(m,n)|$. \label{tab:runtime2}}
\end{table}

We observe two trends: 
\begin{itemize}
    \item 
    Algorithm~\ref{alg:improvedfinal} runs faster on full-rank matrices for $r \in \{2,3\}$, especially as the dimension increases.  For example, it takes Algorithm~\ref{alg:improvedfinal} about 8.5 seconds to process $10^4\times10^4$ generic matrices of rank 2, while 0.14 seconds to certify no rank-2 solution exist for full rank generic matrices of the same size. The reason is that when a feasible solution exists in the low-rank case, it has to construct it and go through all $m$ rows of $W$ and $n$ columns of $H$. 
    In the full-rank generic case, even an $(r+1)\times(r+1)$ submatrix will not have a feasible rank-$r$ solution, and hence the runtime of the algorithm is almost independent of $(m,n)$. It will not need to explore more than $r+1$ rows and columns of $W$ and $H$ to decide for infeasibility.  

    \item For more difficult instances with $r > 3$ and  $m,n$ close to $r$, 
    the full-rank matrices take typically more time to process. Algorithm~\ref{alg:improvedfinal} has to try all sign patterns for a given $(I,J)$ to make sure there is no feasible solution. It compensates for the fact that it may exit the loop over the set $P$ earlier as observed above. 
    While for low-rank matrices, it will stop as soon as a feasible solution is found. 
    This also explains why the standard deviation is close to zero for full-rank matrices: all sign patterns have to be checked since none are feasible. On the contrary, for low-rank matrices, there is randomness on how many patterns have to be tried before finding a feasible one. For example, for $r=5$, $m=n=7$, the shortest runtime for low-rank matrices over the 20 runs is 3.5 s and the longest is 315 s, while there are all about 286 s in the full-rank case. 
\end{itemize}

\subsection{Complexity of ExactEPMF} 
\label{sec:complexexactEPMF}

We can now prove that solving \textsc{ExactEPMF} for $r$ fixed can be done in polynomial time, while, for the generic case when $r$ is part of the input, the algorithm is FPT in the parameter $r$. 

\begin{corollary}\label{cor:abs-p-root-rank}
Let {$p$ be a positive integer},
$X\in\mathbb R_{\geq0}^{m\times n}$, and $r\geq 1$ be fixed. There is a deterministic algorithm running in time polynomial in $m$ and $n$ that decides whether there exist matrices
$W\in\mathbb R^{m \times r}$ and $H\in\mathbb R^{r \times n}$ such that $X=|WH|^{\circ p}$, and outputs a solution if such a pair $(W,H)$ exists. 
Moreover, when $r$ is a part of the input, the algorithm is FPT in $r$ for generic input matrices. 
\end{corollary}

\begin{proof} 
Set\footnote{Except for this step, the rest of the algorithm can be performed in the Blum-Shub-Smale (BSS) model of computation~\cite{BSS89,BCSS98}. To accommodate this step, we require oracle access to exact $p$th root computation for a fixed $p$.} $M=X^{\circ 1/p}$. There exist $W\in\mathbb R^{m\times k}$ and $H\in\mathbb R^{k\times n}$ for some $k\leq r$ such that $X=|WH|^{\circ p}$ if and only if there exists a matrix $Y\in\mathbb R^{m\times n}$ such that
\[
|Y|=M
\qquad\text{and}\qquad
\operatorname{rank}(Y)\leq r.
\]
Indeed, if $Y=WH$, then $|Y|^{\circ p}=X$.
In the other way, if such a matrix $Y$ exists, take any rank-$k$ factorization $Y=WH$ with $W\in\mathbb R^{m\times k}$ and $H\in\mathbb R^{k\times n}$.
Thus, this problem is exactly LRMS for the positive matrix
$M=X^{\circ 1/p}$, and the result follows from Theorems~\ref{thm:fixed-rank-signing} and~\ref{thm:FPTgeneric}.
\end{proof}


The next section shows that this tractability of EPMF is specific to the exact case: once one considers the Frobenius-norm approximation problem, the fixed-rank problem becomes NP-hard already for $r=2$.

\section{{Strong} NP-hardness of FroEPMF with \texorpdfstring{$r$}{r} fixed}\label{sec:fro-fixed-rank}


We now turn to the Frobenius-norm version of EPMF in the fixed-rank regime. 
Given a nonnegative matrix $X\in\mathbb R_{\geq 0}^{m\times n}$ and an integer $r$, the Frobenius absolute $p$th power factorization problem is 
\[
\min_{W\in\mathbb{R}^{m\times r}, H\in\mathbb{R}^{r\times n}}\left\|X-|WH|^{\circ p}\right\|_F^2.
\]

\paragraph{The rank-one case.}
As for NMF, the rank-one case is elementary and can be solved in polynomial time {in the BSS model, given oracle access to exact $p$th root computation for a fixed $p$}.
Indeed, when $r=1$, writing $W=w\in\mathbb R^m$ and $H^\top=h\in\mathbb R^{n}$, we have
\[
|WH|^{\circ p}_{ij} = |w_i h_j|^p = |w_i|^p |h_j|^p.
\]
Thus, by setting $a_i=|w_i|^p$ and $b_j=|h_j|^p$, the rank-one \textsc{FroEPMF} problem is equivalent to
\[
\min_{a\in\mathbb R_+^m,\ b\in\mathbb R_+^n} \|X-ab^\top\|_F^2 .
\]
Since $X\geq 0$, the best unconstrained rank-one approximation of $X$ in Frobenius norm admits a nonnegative solution. 
The rank-one \textsc{FroEPMF} problem can be solved in polynomial time using the truncated SVD: let $(u,\sigma,v)$ be the first singular triplets of $X$, then $a = \sigma |u|$ and $b = |v|$ is an optimal nonnegative solution, because 
$(X(i,j) -  \sigma u_i v_j)^2 \geq 
(X(i,j) -  \sigma |u_i| |v_j|)^2$ 
since $X \geq 0$. 
Note that it could happen that $a_i = 0$ for some $i$, e.g., when $X(i,:) = 0$ (and similarly for $b_j$).  

\paragraph{The rank-two case.} 
The main result of this section shows that the tractability of the rank-one case disappears immediately when $r=2$.
Our proof proceeds by a polynomial-time reduction from the decision version of the NP-hard \textsc{Cut-Norm} problem for sign matrices~\cite{gillis2018complexity}. 
We begin by formally stating the decision versions of both problems.

\noindent
\fbox{\begin{minipage}{0.99\textwidth}\textsc{Cut-Norm.}

\noindent
\emph{Given:} a matrix $M\in\{\pm 1\}^{s\times t}$, and $D\in\mathbb{N}$.

\noindent
\emph{Question:} does there exist $(u,v)\in\{\pm 1\}^s\times\{\pm 1\}^t$ such that $\langle M,uv^\top\rangle \geq D$?
\end{minipage}}\\

\medskip
\noindent
\fbox{\begin{minipage}{0.99\textwidth}\textsc{FroEPMF.} 

\noindent
\emph{Given:} a matrix $X\in\mathbb{R}_{\geq 0}^{m\times n}$ and $T>0$.

\noindent
\emph{Question:} for fixed $p,r\geq 1$, does there exist $(W,H)\in\mathbb{R}^{m\times r}\times\mathbb{R}^{r\times n}$ such that $\|X-|WH|^{\circ p}\|_F^2\leq T$?
\end{minipage}}\\

\textsc{Cut-Norm} is strongly NP-hard due to a strongly polynomial-time reduction from the \textsc{Max-Cut} problem~\cite[Theorem~1]{gillis2018complexity}. 

\paragraph{Reduction.}Let us now reduce \textsc{Cut-Norm} to \textsc{FroEPMF} with $r=2$ in polynomial time.
We may assume $D\leq st$, since otherwise the \textsc{Cut-Norm} instance is trivially a no-instance and can be mapped to any fixed no-instance of rank-2 \textsc{FroEPMF}.

Given an instance $(M,D)$ of \textsc{Cut-Norm}, we construct an instance $(X,T)$ of rank-2 \textsc{FroEPMF} as follows: 
\begin{itemize}
    \item $N=\left\lceil\max\left\{(24p)^2(st)^3,(24p)^p(st)^{p+\frac{1}{2}}
        \right\}\right\rceil$, 
    \item $T=\frac{st-D}{2}$, 
    \item $m=s+2N$ and $n=t+2N$, 
    \item  $X\in\{0,1,2^p\}^{m\times n}$ is defined as follows:
\begin{equation}\label{eq:reductionFroEPMF}
X=
\begin{pmatrix}
B & \mathbf 1_{s\times N} & \mathbf 1_{s\times N}\\[1mm]
\mathbf 1_{N\times t} &  2^p \, \mathbf 1_{N\times N} & \mathbf 0_{N\times N}\\[1mm]
\mathbf 1_{N\times t} & \mathbf 0_{N\times N} &  2^p \, \mathbf 1_{N\times N}
\end{pmatrix},    
\end{equation}

with 
$$B=\frac{M+ \mathbf 1_{s\times t}}{2}\in\{0,1\}^{s\times t}. $$
\end{itemize}

\begin{theorem}\label{thm:red_CUT_to_R2FroEPMF}
The instance $(X,T)$ is a yes-instance of rank-2 \textsc{FroEPMF} if and only if $(M,D)$ is a yes-instance of \textsc{Cut-Norm}. 
Hence, rank-2 \textsc{FroEPMF} is NP-hard. 
\end{theorem}

\begin{proof}
Let us denote 
\[
I=\{1,\dots,s\},\quad
R=\{s+1,\dots,s+N\},\quad
R'=\{s+N+1,\dots,s+2N\}, 
\]
 the indices of the three row blocks of $X$, and
\[
J=\{1,\dots,t\},\quad
C=\{t+1,\dots,t+N\},\quad
C'=\{t+N+1,\dots,t+2N\}, 
\]
 the indices of the three column blocks of $X$.

\textbf{If part.}
Suppose first that $(M,D)$ is a yes-instance of \textsc{Cut-Norm}.
It means that there exist $u\in\{\pm1\}^s$ and
$v\in\{\pm1\}^t$ such that $u^\top Mv\geq D$.
We define $W\in\mathbb{R}^{(s+2N)\times 2}$ and $H\in\mathbb{R}^{2\times (t+2N)}$ as follows:
\[
W(i,:)=
\begin{cases}
(1,\,u_i), & i\in I,\\
(2,\,0), & i\in R,\\
(0,\,2), & i\in R',
\end{cases}
\qquad
H(:,j)=
\begin{cases}
\tfrac{1}{2}(1,v_j)^\top, & j\in J,\\[1mm]
(1,0)^\top, & j\in C,\\[1mm]
(0,1)^\top, & j\in C'.
\end{cases}
\]
In order to show that $\|X-|WH|^{\circ p}\|_F^2 \leq T$, we examine the nine blocks of $X$: 
\begin{itemize}
    \item $I\times J$: 
    For $(i,j)\in I\times J$, we have $W(i,:)=(1,u_i)$ and $H(:,j)=\tfrac{1}{2}(1,v_j)^\top$ which gives
    \[
    W(i,:)H(:,j)=\frac{u_iv_j+1}{2}\in\{0,1\}.
    \]
    The objective function on $X(I,J)$ is 
\begin{eqnarray*}
\|X(I,J)-|W(I,:)H(:,J)|^{\circ p}\|_F^2 & = & \sum_{i=1}^s\sum_{j=1}^t\left(\frac{M_{ij}+1}{2}-\frac{u_iv_j+1}{2}\right)^2\\
& = & \frac{1}{4}\sum_{i=1}^s\sum_{j=1}^t M_{i,j}^2-2u_iv_jM_{ij}+(u_iv_j)^2\\
& = & \frac{1}{4}\left(2st-2\sum_{i=1}^s\sum_{j=1}^tu_iv_jM_{ij}\right)\\
& = & \frac{st-\langle M,uv^\top \rangle}{2}.
\end{eqnarray*}

\item $R\times C$ and $R'\times C'$: the matrix $|WH|^{\circ p}$ is equal to $2^p$, hence
    \[\|X(R,C)-|W(R,:)H(:,C)|^{\circ p}\|_F^2=\|X(R',C')-|W(R',:)H(:,C')|^{\circ p}\|_F^2=0.\]

    \item $R\times C'$ and $R'\times C$: the matrix $|WH|^{\circ p}$ is equal to $0$, hence
    \[\|X(R,C')-|W(R,:)H(:,C')|^{\circ p}\|_F^2=\|X(R',C)-|W(R',:)H(:,C)|^{\circ p}\|_F^2=0.\]
    
    \item $I\times C$, $I\times C'$, $R\times J$, $R'\times J$: every inner product has absolute value~$1$, so $|{\cdot}|^p=1$, matching the entries of $X$ exactly.
\end{itemize} 
Finally,
\[
\|X-|WH|^{\circ p}\|_F^2 = \frac{st-\langle M,uv^\top \rangle}{2} \leq \frac{st-D}{2} = T.
\]
Hence $(X,T)$ is a yes-instance of rank-2 \textsc{FroEPMF}.

\medskip
\noindent
\textbf{Only if part.}
Suppose now that $(X,T)$ is a yes-instance of rank-2 \textsc{FroEPMF}, that is, there exist
$W\in\mathbb{R}^{m\times 2}$ and $H\in\mathbb{R}^{2\times n}$ such that
\begin{equation}\label{eq:hypo-onlyif}
    \|X-|WH|^{\circ p}\|_F^2\leq T.
\end{equation}
We show that this implies that $(M,D)$ is a yes-instance of \textsc{Cut-Norm}.
To do so, we exploit the repeated rows and columns of $X$.
For a fixed $H$, all rows
$X(i,:)$ with $i\in R$ are the same; hence, we may choose all the rows $W(i,:)$, $i\in R$, to be equal in an optimal solution.
The same argument applies to the rows indexed by $R'$, and symmetrically to the columns indexed by $C$ and $C'$.
Hence, without loss of generality, we may assume that $(W,H)$ has the form
\begin{equation}\label{eq:solonlyif}
W=
\begin{pmatrix}
a&a'\\
r_1\mathbf 1_N & r_2\mathbf 1_N\\
r'_1\mathbf 1_N & r'_2\mathbf 1_N
\end{pmatrix},
\qquad
H=
\begin{pmatrix}
b&b'\\
c_1\mathbf 1_N&c_2\mathbf 1_N\\
c'_1\mathbf 1_N&c'_2\mathbf 1_N
\end{pmatrix}^\top,
\end{equation}
for some $a,a'\in\mathbb{R}^s$, $b,b'\in\mathbb{R}^t$, and
$r_1,r'_1,r_2,r'_2,c_1,c'_1,c_2,c'_2\in\mathbb{R}$.
From this solution of the rank-2 \textsc{FroEPMF} instance, we construct a solution of the \textsc{Cut-Norm} instance.
We define $u\in\{\pm1\}^s$ and $v\in\{\pm1\}^t$ using the signs of $a$ and $b$:
\begin{equation}\label{eq:def_sigma_tau}
u_i=
\begin{cases}
+1,& \text{when~~} a_ia'_i\geq 0,\\
-1,& \text{when~~} a_ia'_i<0,
\end{cases}
\quad
v_j=
\begin{cases}
+1,& \text{when~~} b_jb'_j\geq 0,\\
-1,& \text{when~~} b_jb'_j<0.
\end{cases}
\end{equation}
It remains to prove that with this solution $(u,v)$ and the hypothesis \eqref{eq:hypo-onlyif}, we have $\langle M,uv^\top \rangle\geq D$.

We now decompose the objective according to four parts of $X$, following the division $[s+2N]=I\cup \widetilde{R}$ and $[t+2N]=J\cup \widetilde{C}$, where $\widetilde{R}=R\cup R'$ and $\widetilde{C}=C\cup C'$:
\[
\|X-|WH|^{\circ p}\|_F^2 = E_{IJ}+E_{I\widetilde{C}}+E_{\widetilde{R}J}+E_{\widetilde{R}\widetilde{C}},
\]
where $E_{IJ}$ is the contribution of the top-left block $I\times J$,
$E_{I\widetilde{C}}$ is the contribution of the top-right blocks $I\times \widetilde{C}$, 
$E_{\widetilde{R}J}$ the contribution of the bottom-left blocks $\widetilde{R}\times J$,
and $E_{\widetilde{R}\widetilde{C}}$ the contribution of the bottom-right blocks
$\widetilde{R}\times \widetilde{C}$. Note that since
\[ T=\frac{st-D}{2}\leq st, \]
we have $\|X-|WH|^{\circ p}\|_F^2\leq T \leq st$.

\medskip
\noindent\textbf{Step 1: {the $\widetilde{R} \times \widetilde{C}$ bottom-right block} constrains $r_1,r_1',r_2,r_2',c_1,c_1',c_2,c_2'$.}  
We first analyze the four $N\times N$ blocks in the bottom-right corner. Their
contribution is
\begin{eqnarray}
E_{\widetilde{R}\widetilde{C}} & = &  \|X(\widetilde{R},\widetilde{C})-|W(\widetilde{R},:)H(:,\widetilde{C})|^{\circ p}\|_F^2 \notag\\
&=&
N^2\Bigl(
(2^p-|r^\top c|^p)^2
+|r^\top c'|^{2p}
+|r'^\top c|^{2p}
+(2^p-|r'^\top c'|^p)^2
\Bigr).
\label{eq:EBR}
\end{eqnarray}
where $r = [r_1,r_2]$, and similarly for $r'$ and $c$. Now suppose $c'=\lambda c$ and fix $x=|r^\top c|^p$, $y=|r'^\top c|^p$, hence, {using Lemma~\ref{lem:abxy_inequality} with $a=2^p$ and $b=|\lambda|^p$,}  
\[ 
E_{\widetilde{R}\widetilde{C}}  =   N^2\Bigl(
(2^p-x)^2
+|\lambda|^{2p}x^{2}
+y^{2}
+(2^p-|\lambda|^{p}y)^2
\Bigr)
 \geq  2^{2p}N^2, 
\] 
a contradiction since $E_{\widetilde{R}\widetilde{C}}\leq \|X-|WH|^{\circ p}\|_F^2\leq T \leq st$. Since $WH=(WQ)(Q^{-1}H)$ for any invertible matrix $Q$, we may  assume $c=(1,0)$ and
$c'=(0,1)$.
Therefore 
\begin{eqnarray}
E_{\widetilde{R}\widetilde{C}} & = &
N^2\bigl((2^p-r_1^p)^2+|r_2|^{2p}+|r_1'|^{2p}+(2^p-{r_2'}^p)^2\bigr). 
\label{eq:EBR3}
\end{eqnarray} 
Since $E_{\widetilde R\widetilde C}\leq st$ and all four terms in parentheses in \eqref{eq:EBR3} are nonnegative, each of these terms is at most $st/N^2$. 
Therefore, by denoting $\varepsilon=\sqrt{st}/N<1$, we have $r_1^p\in[2^p-\varepsilon,2^p+\varepsilon]$ and $r_1\geq 1$ since $r_1^p\ge 2^p-1\ge 1$ (the same holds for $r_2'$).
Moreover, for $\alpha,\beta\geq 1$, we have $|\alpha^p-\beta^p|\geq |\alpha-\beta|$. 
Thus,
\[
        r_1\in[2-\varepsilon,2+\varepsilon],
        \qquad
        r'_2\in[2-\varepsilon,2+\varepsilon].
\]
Since $\varepsilon<1$, we also have $\varepsilon\leq \varepsilon^{1/p}$.
For simplicity, we use the weaker but uniform following bounds:
\begin{equation}\label{eq:boundparameters}
r_1,~r'_2\in\left[2-\varepsilon^{\frac{1}{p}},2+\varepsilon^{\frac{1}{p}}\right],
~~~
r'_1,~r_2\in\left[-\varepsilon^{\frac{1}{p}},\varepsilon^{\frac{1}{p}}\right].
\end{equation}

\medskip
\noindent\textbf{Step 2: {the $I\times\widetilde{C}$ and $\widetilde{R}\times J$ side blocks} constrain the entries of $a,a',b,b'$.}  
We analyze the two rectangular blocks $I\times\widetilde{C}$:  
\begin{eqnarray*}
E_{I\widetilde{C}}
& = & \|X(I,C\cup C')-|W(I,:)H(:,C\cup C')|^{\circ p}\|_F^2 
 =  N\sum_{i=1}^s
\Bigl((1-|a_i|^p)^2+(1-|a'_i|^p)^2\Bigr). 
\end{eqnarray*} 
Since $E_{I\widetilde{C}}\leq st$, $\bigl||a_i|^p-1\bigr|\leq \sqrt{\frac{st}{N}}, \qquad \bigl||a'_i|^p-1\bigr|\leq \sqrt{\frac{st}{N}}$, that is, 
\begin{equation}\label{eq:boundsa1a2}
        \bigl||a_i|-1\bigr|\leq \sqrt{\frac{st}{N}}, \qquad \bigl||a'_i|-1\bigr|\leq \sqrt{\frac{st}{N}}. 
\end{equation} 
 In the same way, for the blocks $R\times J$ and $R'\times J$, 
\[
  \bigl||r_1 b_j+r_2 b_j'|^p-1\bigr|\leq \sqrt{\frac{st}{N}}, \quad \text{which gives } \bigl||r_1 b_j+r_2b_j'|-1\bigr|\leq \sqrt{\frac{st}{N}}. 
\]
and
\[
  \bigl||r'_1 b_j+r'_2 b_j'|^p-1\bigr|\leq \sqrt{\frac{st}{N}}, \quad \text{which gives } \bigl||r'_1  b_j+r'_2 b_j'|-1\bigr|\leq \sqrt{\frac{st}{N}},
\]
Using these relations and the bounds \eqref{eq:boundparameters}, we obtain, by Lemma~\ref{lem:bounds05} (see Appendix~\ref{sec:bounds05}), the following bounds on $|b_j|$ and $|b'_j|$:
\begin{equation}\label{eq:boundsb1b2}
        \left| |b_j|-\frac{1}{2} \right|\leq \sqrt{N}\varepsilon+\varepsilon^{1/p}, \qquad \left||b'_j|-\frac{1}{2}\right|\leq \sqrt{N}\varepsilon+\varepsilon^{1/p}.
\end{equation}

\medskip
\noindent\textbf{Step 3: the {$I\times J$} top-left block controls the disagreements between $M$ and $uv^\top$.} 
We now turn to the block $I\times J$.
We distinguish the two types of disagreements
between $M$ and $uv^\top$, in the sets $F_0$ and $F_1$ defined below, together with the two types of agreeing entries, in the sets $F_2$ and $F_3$: 
\begin{eqnarray*}
F_0&=&\{(i,j)\in I\times J:\ M_{ij}=-1,\ u_iv_j=1\},\\
F_1&=&\{(i,j)\in I\times J:\ M_{ij}=1,\ u_iv_j=-1\},\\
F_2&=&\{(i,j)\in I\times J:\ M_{ij}=1,\ u_iv_j=1\},\\
F_3&=&\{(i,j)\in I\times J:\ M_{ij}=-1,\ u_iv_j=-1\}.
\end{eqnarray*}
The sets $F_0$ and $F_1$ are exactly the disagreement positions between $M$ and $uv^\top$, and we have 
\begin{equation}\label{eq:disagreement_inner_product}
\langle M,uv^\top\rangle = st-2(|F_0|+|F_1|).
\end{equation}
In the only-if direction, the factors~\eqref{eq:solonlyif} are supposed to be arbitrary.
Thus, unlike in
the construction of the if direction, we cannot assume that
$\tilde{a}_i:=(a_i,a'_i)^\top=(1,u_i)^\top$ and $\tilde{b}:=(b_j,b'_j)^\top=\frac{1}{2}(1,v_j)^\top$ on the block $I\times J$.
However, we have seen that the side blocks force an approximate version of this structure:
$|a_i|$ and $|a'_i|$ are close to $1$, while
$|b_j|$ and $|b'_j|$ are close to $1/2$.
Hence, the signs defined in
\eqref{eq:def_sigma_tau} allow us to link the values
$|\tilde{a}_i^\top \tilde{b}_j|^p$ to the objective~\eqref{eq:disagreement_inner_product}.
In the following, we decompose $E_{IJ}$ over the four sets and remove the nonnegative contributions coming from the agreeing positions:
\[
\begin{aligned}
E_{IJ}
&=
\sum_{(i,j)\in F_0} (0-|\tilde{a}_i^\top \tilde{b}_j|^p)^2
+
\sum_{(i,j)\in F_1} (1-|\tilde{a}_i^\top \tilde{b}_j|^p)^2
+
\sum_{(i,j)\in F_2} (1-|\tilde{a}_i^\top \tilde{b}_j|^p)^2
+
\sum_{(i,j)\in F_3} (0-|\tilde{a}_i^\top \tilde{b}_j|^p)^2
\\
&\geq
\sum_{(i,j)\in F_0} |\tilde{a}_i^\top \tilde{b}_j|^{2p}
+
\sum_{(i,j)\in F_1} (1-|\tilde{a}_i^\top \tilde{b}_j|^p)^2.
\end{aligned}
\]
Now we have to bound $|\tilde{a}_i^\top \tilde{b}_j|^p$ on $F_0$ and $F_1$.
By denoting $\eta:=\sqrt{\frac{st}{N}}+\left(\frac{\sqrt{st}}{N}\right)^{\frac{1}{p}}$ and using \eqref{eq:boundsa1a2},\eqref{eq:boundsb1b2}, we have 
\begin{equation}\label{eq:finalboundsab}
    \big| |a_i|-1\big| \leq \eta, \quad \left||b_j|-\frac{1}{2}\right|\leq \eta, \quad \big| |a'_i|-1\big| \leq \eta, \quad \text{ and } \quad \left||b'_j|-\frac{1}{2}\right|\leq \eta.
\end{equation}
Since $N\geq (24p)^2(st)^3$, 
\[
        \sqrt{\frac{st}{N}}
        \le
        \sqrt{\frac{st}{(24p)^2(st)^3}}
        =
        \frac{1}{24pst}.
\]
Moreover, since $N\ge (24p)^p(st)^{p+\frac12}$, we get
\[
        \left(\frac{\sqrt{st}}{N}\right)^{1/p}
        \le
        \left(
        \frac{\sqrt{st}}{(24p)^p(st)^{p+\frac{1}{2}}}
        \right)^{1/p}
        =
        \frac{1}{24pst}.
\]
Therefore
\begin{equation}\label{eq:boundetaonlypst}
        \eta
        =
        \sqrt{\frac{st}{N}}
        +
        \left(\frac{\sqrt{st}}{N}\right)^{1/p}
        \le
        \frac{1}{12pst}.
\end{equation}

Let us analyze the indices in $F_0$ and $F_1$ separately: 
\begin{itemize}
    \item For $(i,j)\in F_0$, $u_iv_j=1$, hence both terms $a_ib_j$ and $a'_ib'_j$ have the same sign.
    Thus  
    \[
    |\tilde{a}_i^\top \tilde{b}_j| = |a_ib_j+a'_ib'_j|=|a_ib_j|+|a'_ib'_j|\geq 2\left(1-\eta\right)\left(\frac{1}{2}-\eta\right)\geq 1-3\eta.
    \]
    By combining this bound on $|\tilde{a}_i^\top \tilde{b}_j|$ with~\eqref{eq:boundetaonlypst}, 
    \begin{equation}\label{eq:boundF0}
    |\tilde{a}_i^\top \tilde{b}_j|^p \geq (1-3\eta)^p\geq 1-3p\eta\geq 1 - \frac{1}{4st}.
    \end{equation}
    
    \item For $(i,j)\in F_1$, $u_iv_j=-1$; hence both terms $a_ib_j$ and $a'_ib'_j$ have opposite signs.
    Thus 
    \[
    |\tilde{a}_i^\top \tilde{b}_j| = ||a_ib_j|-|a'_ib'_j||\leq (1+\eta)\left(\frac{1}{2}+\eta\right)-(1-\eta)\left(\frac{1}{2}-\eta\right)=3\eta.
    \]
    By combining this bound on $|\tilde{a}_i^\top \tilde{b}_j|$ with~\eqref{eq:boundetaonlypst}, 
    \begin{equation}\label{eq:boundF1}
    |\tilde{a}_i^\top \tilde{b}_j|^p \leq (3\eta)^p\leq \left(\frac{1}{4pst}\right)^p\leq \frac{1}{4st}.
    \end{equation}
\end{itemize}
By using~\eqref{eq:boundF0} and \eqref{eq:boundF1}, we obtain
\[
\begin{aligned}
E_{IJ}
&\geq
\sum_{(i,j)\in F_0} |\tilde{a}_i^\top \tilde{b}_j|^{2p}
+
\sum_{(i,j)\in F_1} (1-|\tilde{a}_i^\top \tilde{b}_j|^p)^2\\
&\geq
\left(1-\frac{1}{4st}\right)^2
|F_0|+\left(1-\frac{1}{4st}\right)^2
|F_1|\\
&\geq
\left(1-\frac{1}{4st}\right)^2
\bigl(|F_0|+|F_1|\bigr) \geq
\left(1-\frac{1}{2st}\right)
\bigl(|F_0|+|F_1|\bigr).
\end{aligned}
\]
Since $E_{IJ}\leq T$, $T\leq \frac{st}{2}$ and $st\geq 1$, 
\[
        |F_0|+|F_1|
        \leq
        \frac{T}{1-\frac1{2st}}=
        T+\frac{\frac{T}{2st}}{1-\frac1{2st}}
        \leq
        T+\frac{\frac{1}{4}}{1-\frac{1}{2}}
        = T+\frac{1}{2}.
\]
Hence
\[
        \langle M,uv^\top\rangle = st-2(|F_0|+|F_1|) > st - 2(T+\frac{1}{2}) \geq st - 2\frac{st-D}{2}-1=D-1.
\]
Since $\langle M,uv^\top\rangle$ and $D$ are integers, this implies $\langle M,uv^\top\rangle\geq D$.
\end{proof}

\paragraph{Discussion.} As in Section~\ref{sec:strongNPhardness}, the above reduction, see \eqref{eq:reductionFroEPMF}, uses only a fixed number of distinct values in the input matrix $X \in \{0,1,2^p\}^{m \times n}$. It remains open whether the same hardness result holds for binary input matrices. 
We leave this as a direction of future research.

\section{Conclusion} 
\label{sec:concl}

In this paper, we studied the computational complexity of entrywise power matrix factorization (EPMF), a class of non-linear matrix decompositions (NMDs). 
We showed that EPMF is tractable in the exact fixed-rank regime, but otherwise it is NP-hard. 
For the exact fixed rank regime, we showed that the problem is equivalent to a low-rank matrix signing (LRMS) problem: flip the signs of some entries of a given matrix to make it rank $r$. We proposed a constructive algorithm polynomial in $(m,n)$ to solve this problem (Theorem~\ref{thm:fixed-rank-signing}). 
Moreover, we showed that the algorithm is FPT in the parameter $r$ for generic matrices (Theorem~\ref{thm:FPTgeneric}). 
When the rank is not fixed, the problem is strongly NP-hard in the exact case (Corollary~\ref{thm:exactEPMF-strongly-nphard}), and the least-squares optimization problem is already strongly NP-hard as well for the smallest nontrivial rank, namely $r=2$ (Theorem~\ref{thm:red_CUT_to_R2FroEPMF}). 

Further works include the study of conditions under which EPMF, and other NMDs, are tractable (e.g., using convexifications, or analyzing the loss landscape, as done for many low-rank matrix approximation problems), as well as the design of efficient algorithms and their use in real-world applications. 
We could also explore whether we can further reduce the computational cost of the algorithm that solves the sign problem in Theorem~\ref{thm:fixed-rank-signing}.

\paragraph*{Use of AI assistance.} During the preparation of this manuscript, the authors used ChatGPT (OpenAI) to assist with improving the presentation of several proofs, exploring alternative proof strategies, and checking the clarity and logical consistency of mathematical arguments. The AI system was used solely as a writing and reasoning aid; all mathematical results, proofs, and conclusions were independently verified by the authors, who take full responsibility for the content of this manuscript.

{\small
\bibliographystyle{spmpsci}
\bibliography{bib.bib}}

\newpage

\appendix
\section{Lemmas for the bounds on the coefficients around \texorpdfstring{$1/2$}{1/2}}
\label{sec:bounds05}

\begin{lemma}[Bounds on the coefficients around $1/2$]\label{lem:bounds05}

\noindent
Let $\varepsilon>0$, $p\geq 1$, $N\geq 1$ such that $\varepsilon^{\frac{1}{p}}< \frac{1}{2}$, and define $\delta:=\sqrt{N}\varepsilon, \eta:=\varepsilon^{1/p}$. Let $x_1,x_2\in\mathbb{R}$ be such that
\[
\big||(2+a)x_1+bx_2|-1\big|\leq \delta,
\qquad
\big||cx_1+(2+d)x_2|-1\big|\leq \delta, 
\]
where $a,b,c,d\in[-\eta,\eta]$. Then, for $i=1,2$, 
\[
\left||x_i|-\frac{1}{2}\right|
\leq
\delta+\eta, \quad \text{that is}, \quad 
\left||x_i|-\frac{1}{2}\right|
\leq
\sqrt{N}\varepsilon+\varepsilon^{1/p}.
\]
\end{lemma}

\begin{proof}
From the assumptions, there exist signs $\sigma_1,\sigma_2\in\{-1,1\}$ and errors
$e_1,e_2$ with $|e_1|,|e_2|\leq \delta$ such that
\[
(2+a)x_1+bx_2=\sigma_1+e_1,
\quad \text{and} \quad cx_1+(2+d)x_2=\sigma_2+e_2,
\]
which gives
\[
2\left(x_1-\frac{\sigma_1}{2}\right) = e_1-ax_1-bx_2
\quad \text{and} \quad 2\left(x_2-\frac{\sigma_2}{2}\right) = e_2-dx_2-cx_1.
\]
By using the bounds on the different elements, we obtain
\begin{equation}\label{eq:twoeqbounds}
2\left|x_1-\frac{\sigma_1}{2}\right|
\leq
\delta+\eta |x_1|+\eta |x_2|
\quad \text{and} \quad
2\left|x_2-\frac{\sigma_2}{2}\right|
\leq
\delta+\eta |x_2|+\eta |x_1|.
\end{equation}
Set
\[
M:=\max_{i=1,2}\left|x_i-\frac{\sigma_i}{2}\right| \quad \text{and} \quad B:=\max_{i=1,2}\left|x_i\right| \quad \text{for which} \quad B\leq \frac{1}{2}+M.
\]
Using both equations of \eqref{eq:twoeqbounds}, we obtain $2M \leq \delta + 2\eta B \leq \delta + \eta + 2\eta M$, hence, $(2-2\eta)M
\leq
\delta+\eta$.
\noindent
Since $\eta< 1/2$, for $i=1,2$,
\[
    \left||x_i|-\frac{1}{2}\right|
 \leq  \left|x_i-\frac{\sigma_i}{2}\right| \leq M \leq \frac{\delta+\eta}
{2-2\eta}\leq \delta+\eta.
\] 
\end{proof}

\begin{lemma}\label{lem:abxy_inequality}
    Let $a,b\in \mathbb{R}$. Then
    \[
        a^2 \; = \; \min_{x,y\in \mathbb{R}} (a-x)^2+b^2x^2+y^2+(a-by)^2.
    \]
    \begin{proof}
        Let $f(x,y)=(a-x)^2+b^2x^2+y^2+(a-by)^2$. Then
        \[
         \nabla f(x,y)=\left(\begin{array}{c}
             -2(a-x)+2b^2x \\
             2y-2b(a-by) \\ 
         \end{array}\right)=
         \left(\begin{array}{c}
             (2+2b^2)x-2a \\
             (2+2b^2)y-2ab \\ 
         \end{array}\right), \qquad
         \nabla^2 f(x,y)=\left(\begin{array}{cc}
             2+2b^2 & 0 \\
             0 & 2+2b^2 \\ 
         \end{array}\right),
        \]
        which means that $f$ is strictly convex and attains the minimum when $\nabla f=0$, that is at $x_\star=\frac{a}{1+b^2}$ and $y_\star=\frac{ab}{1+b^2}$. Hence 
        \[
         f(x_\star,y_\star)=\left(a-\frac{a}{1+b^2}\right)^2+\frac{a^2b^2}{(1+b^2)^2}+\frac{a^2b^2}{(1+b^2)^2}+\left(a-\frac{ab^2}{1+b^2}\right)^2=\frac{a^2b^4+2a^2b^2+a^2}{(1+b^2)^2}=a^2.
        \]
    \end{proof}
\end{lemma}

\end{document}